\makeatletter
\def\input@path{{styles/}{../styles/}} 
\makeatother

\ifx\SubmitVerFlag\undefined
\newcommand{\SubmitVer}[1]{}
\newcommand{\NotSubmitVer}[1]{#1}
\else
\newcommand{\SubmitVer}[1]{#1}
\newcommand{\NotSubmitVer}[1]{}
\fi

\SubmitVer{

}

\ifx\SubmitVerFlag\undefined
   \documentclass[12pt]{article}%
      \AtBeginDocument{%
         \providecommand\BibTeX{{%
               \normalfont B\kern-0.5em{\scshape i\kern-0.25em
                  b}\kern-0.8em\TeX}}}%
   
\else
   \documentclass[acmsmall,review,screen]{acmart}%

   \AtBeginDocument{%
      \providecommand\BibTeX{{%
            \normalfont B\kern-0.5em{\scshape i\kern-0.25em
               b}\kern-0.8em\TeX}}}%

   \setcopyright{acmcopyright}%
   \copyrightyear{2018}%
   \acmYear{2018}%
   \acmDOI{10.1145/1122445.1122456}

   \acmJournal{TALG}%
   \acmVolume{0}%
   \acmNumber{0}%
   \acmArticle{111}


\fi

\providecommand{\BibLatexMode}[1]{}

\usepackage{subcaption}%
\NotSubmitVer{
   \usepackage{amsmath}%
   \usepackage{graphicx}%
   \usepackage{amssymb}%
   \usepackage{caption}%
}

\NotSubmitVer{%
   \usepackage[cm]{fullpage}
}
\usepackage{xcolor}%
\usepackage{mleftright}%
\usepackage{scalerel}%
\usepackage{euscript}%
\usepackage{xspace}%
\usepackage{multirow}

\DeclareFontFamily{U}{mathb}{\hyphenchar\font45}
\DeclareFontShape{U}{mathb}{m}{n}{
      <5> <6> <7> <8> <9> <10> gen * mathb
      <10.95> mathb10 <12> <14.4> <17.28> <20.74> <24.88> mathb12
      }{}
\DeclareSymbolFont{mathb}{U}{mathb}{m}{n}
\DeclareFontSubstitution{U}{mathb}{m}{n}

\DeclareMathSymbol{\convolution}   {2}{mathb}{"0A}
\DeclareMathSymbol{\Asterisk}      {2}{mathb}{"06}

\NotSubmitVer{
   \usepackage{titlesec}%
   \titlelabel{\thetitle. }%
   \titleformat{\paragraph}[runin]
   {\normalfont\bfseries}
   {\theparagraph}
   {1em}
   {\addperiod}
   
   \newcommand{\addperiod}[1]{#1.}
   
}

\usepackage{hyperref}%

\providecommand{\IfPrinterVer}[2]{#2}%
\IfPrinterVer{%
   \usepackage{hyperref}%
}{%
   \usepackage{hyperref}%
   \hypersetup{%
      breaklinks,%
      colorlinks=true,%
      urlcolor=[rgb]{0.25,0.0,0.0},%
      linkcolor=[rgb]{0.5,0.0,0.0},%
      citecolor=[rgb]{0,0.2,0.445},%
      filecolor=[rgb]{0,0,0.4},
      anchorcolor=[rgb]={0.0,0.1,0.2}%
   }
}

\NotSubmitVer{%
   \numberwithin{figure}{section}%
   \numberwithin{table}{section}%
   \numberwithin{equation}{section}%
}
\IfFileExists{sariel_computer.sty}{\def\sarielComp{1}}{}
\ifx\sarielComp\undefined%
\newcommand{\SarielComp}[1]{}
\newcommand{\NotSarielComp}[1]{#1}%
\else
\newcommand{\SarielComp}[1]{#1}%
\newcommand{\NotSarielComp}[1]{}%
\fi

\newcommand{\Of}{\Mh{\mathcal{O}}}%

\newcommand{\BSet}{\Mh{{B}}}%
\newcommand{\EBSet}{\Mh{{B^+}}}%

\SubmitVer{%
   \theoremstyle{acmplain}%
   \newtheorem{claim}{Claim}[section]%
}

\NotSubmitVer{

   \usepackage[amsmath,thmmarks]{ntheorem}%
   \theoremseparator{.}%

   \theoremstyle{plain}%
   \newtheorem{theorem}{Theorem}[section]
   
   \newtheorem{lemma}[theorem]{Lemma}
   
   \newtheorem{corollary}[theorem]{Corollary}
   \newtheorem{claim}[theorem]{Claim}%

   \newtheorem{proposition}[theorem]{Proposition}

   \theoremstyle{plain}%
   \theoremheaderfont{\sf} \theorembodyfont{\upshape}%
   \newtheorem*{remark:unnumbered}[FakeCounter]{Remark}%
   \newtheorem*{remarks}[theorem]{Remarks}%
   \newtheorem{remark}[theorem]{Remark}%
   \newtheorem{definition}[theorem]{Definition}
   \newtheorem*{defn:unnumbered}[FakeCounter]{Definition}

   \newcommand{\myqedsymbol}{\rule{2mm}{2mm}}
   
   \theoremheaderfont{\em}%
   \theorembodyfont{\upshape}%
   \theoremstyle{nonumberplain}%
   \theoremseparator{}%
   \theoremsymbol{\myqedsymbol}%
   \newtheorem{proof}{Proof:}%
   
}

\newcommand{\eps}{\varepsilon}%
\newcommand{\epsA}{\Mh{\xi}}

\newcommand{\epsR}{\Mh{\vartheta}}%
\newcommand{\epsRA}{\Mh{\psi}}%

\newcommand{\ceil}[1]{\left\lceil {#1} \right\rceil}

\newcommand{\HLinkShort}[2]{\hyperref[#2]{#1\ref*{#2}}}
\newcommand{\HLink}[2]{\hyperref[#2]{#1~\ref*{#2}}}
\newcommand{\HLinkPage}[2]{\hyperref[#2]{#1~\ref*{#2}%
      $_\text{p\pageref{#2}}$}}
\newcommand{\HLinkPageOnly}[1]{\hyperref[#1]{Page~\refpage*{#1}%
      $_\text{p\pageref{#1}}$}}

\newcommand{\HLinkSuffix}[3]{\hyperref[#2]{#1\ref*{#2}{#3}}}
\newcommand{\HLinkPageSuffix}[3]{\hyperref[#2]{#1\ref*{#2}%
      #3$_\text{p\pageref{#2}}$}}

\newcommand{\seclab}[1]{\label{sec:#1}}
\newcommand{\secref}[1]{\HLink{Section}{sec:#1}}

\newcommand{\corlab}[1]{\label{cor:#1}}
\newcommand{\corref}[1]{\HLink{Corollary}{cor:#1}}%

\providecommand{\deflab}[1]{\label{def:#1}}
\newcommand{\defref}[1]{\HLink{Definition}{def:#1}}

\newcommand{\apndlab}[1]{\label{apnd:#1}}
\newcommand{\apndref}[1]{\HLink{Appendix}{apnd:#1}}

\newcommand{\clmlab}[1]{\label{claim:#1}}
\newcommand{\clmref}[1]{\HLink{Claim}{claim:#1}}

\newcommand{\itemlab}[1]{\label{item:#1}}
\newcommand{\itemref}[1]{\HLinkSuffix{}{item:#1}{}}
\newcommand{\nitemref}[1]{\ref{item:#1}}
\newcommand{\itempref}[1]{\HLinkSuffix{(}{item:#1}{)}}

\newcommand{\remlab}[1]{\label{rem:#1}}
\newcommand{\remref}[1]{\HLink{Remark}{rem:#1}}%
\newcommand{\defrefY}[2]{\hyperref[def:#2]{#1}}

\newcommand{\tbllab}[1]{\label{table:#1}}
\newcommand{\tblref}[1]{\HLink{Table}{table:#1}}

\newcommand{\lemlab}[1]{\label{lemma:#1}}
\newcommand{\lemref}[1]{\HLink{Lemma}{lemma:#1}}%

\newcommand{\proplab}[1]{\label{prop:#1}}
\newcommand{\propref}[1]{\HLink{Proposition}{prop:#1}}%

\newcommand{\thmlab}[1]{{\label{theo:#1}}}
\newcommand{\thmref}[1]{\HLink{Theorem}{theo:#1}}
\newcommand{\nthmref}[1]{Theorem~\ref{theo:#1}}

\providecommand{\eqlab}[1]{}%
\renewcommand{\eqlab}[1]{\label{equation:#1}}
\newcommand{\Eqref}[1]{\HLinkSuffix{Eq.~(}{equation:#1}{)}}

\providecommand{\Mh}[1]{{#1}}%

\IfFileExists{.latex_printer_friendly}{\def\GenPrinterVer{1}}{}%

\ifx\GenPrinterVer\undefined
   \IfFileExists{.latex_color}{\def\GenColorMath{1}}{}
\else
   \renewcommand{\IfPrinterVer}[2]{#1}%
\fi

\ifx\GenColorMath\undefined
\else
\renewcommand{\Mh}[1]{{\textcolor{red}{#1}}}%
\fi

\definecolor{blue25emph}{rgb}{0, 0, 11}
\newcommand{\emphic}[2]{%
   \textcolor{blue25emph}{%
      \textbf{\emph{#1}}}%
   \index{#2}}

\definecolor{almostblack}{rgb}{0, 0, 0.3}
\newcommand{\emphw}[1]{{\textcolor{almostblack}{\emph{#1}}}}%

\newcommand{\emphi}[1]{\emphic{#1}{#1}}%
\newcommand{\emphOnly}[1]{\emph{\textbf{#1}}}%

\SarielComp{%
\definecolor{blue25emph}{rgb}{0, 0, 11}
\renewcommand{\emphic}[2]{%
   \textcolor{blue25emph}{%
      \textbf{\emph{#1}}}%
   \index{#2}}

\renewcommand{\emphOnly}[1]{\emph{\textcolor{blue25emph}{{#1}}}}
}

\usepackage{stmaryrd}%
\newcommand{\IntRange}[1]{\mleft[ #1 \mright]}
\newcommand{\IRX}[1]{\IntRange{#1}}%

\DeclareFontFamily{U}{BOONDOX-calo}{\skewchar\font=45 }
\DeclareFontShape{U}{BOONDOX-calo}{m}{n}{
  <-> s*[1.05] BOONDOX-r-calo}{}
\DeclareFontShape{U}{BOONDOX-calo}{b}{n}{
  <-> s*[1.05] BOONDOX-b-calo}{}
\DeclareMathAlphabet{\mathcalb}{U}{BOONDOX-calo}{m}{n}
\SetMathAlphabet{\mathcalb}{bold}{U}{BOONDOX-calo}{b}{n}
\DeclareMathAlphabet{\mathbcalb}{U}{BOONDOX-calo}{b}{n}

\renewcommand{\Re}{\mathbb{R}}%
\newcommand{\NN}{\mathbb{N}}%

\newcommand{\PS}{\Mh{P}}%
\newcommand{\PA}{\Mh{X}}%
\newcommand{\PB}{\Mh{Y}}%
\newcommand{\PC}{\Mh{Z}}%

\newcommand{\Family}{\Mh{\mathcal{F}}}%
\newcommand{\degY}[2]{\mathrm{deg}\pth{#1, #2}}

\newcommand{\distC}{\Mh{\mathsf{d}}}
\newcommand{\distSetY}[2]{\Mh{\mathsf{d}}\pth{#1,#2}}

\newcommand{\diamX}[1]{\mathrm{diam}\pth{#1}}%
\newcommand{\diamY}[2]{\Mh{\mathrm{diam}}_{#1}\pth{#2}}%

\newcommand{\SSet}{\Mh{S}}%
\newcommand{\Cover}{\Mh{\mathcal{C}}}%

\newcommand{\QS}{\Mh{Q}}%

\newcommand{\Star}{\Mh{\convolution}}
\newcommand{\StarX}[1]{\Star_{#1}}%

\newcommand{\StarR}{\Mh{\Asterisk}}

\newcommand{\Tree}{\Mh{T}}%

\providecommand{\G}{}%
\renewcommand{\G}{\Mh{G}}%
\newcommand{\GS}{\Mh{G_S}}%

\newcommand{\VV}{\Mh{V}}%
\newcommand{\VX}[1]{\VV\pth{#1}}%
\newcommand{\Edges}{\Mh{E}}%
\newcommand{\EdgesX}[1]{\Mh{E}\pth{#1}}%
\newcommand{\EdgesY}[2]{\Mh{E}\pth{#1,#2}}%
\newcommand{\GDef}{\G=(\VV,\Edges)}
\newcommand{\Gnd}{\Mh{\mathcal{G}_{n,d}}}%
\newcommand{\wGnd}{\Mh{\widetilde{\mathcal{G}_{n,d}}}}%

\newcommand{\tT}{\Mh{\widetilde{T}}}%

\newcommand{\CfX}[1]{\Mh{\mathcalb{C}^{}}_{\!#1}}
\newcommand{\cfX}[1]{\Mh{\mathcalb{c}^{}}_{\!#1}}
\newcommand{\cepsR}{\cfX{\epsR}}

\newcommand{\Set}[2]{\left\{ #1 \;\middle\vert\; #2 \right\}}
\newcommand{\cardin}[1]{\left| {#1} \right|}%
\newcommand{\pth}[2][\!]{\mleft({#2}\mright)}%

\newcommand{\DistChar}{\Mh{\mathrm{d}}}%
\newcommand{\dMY}[2]{\DistChar\pth{#1,#2}}%
\newcommand{\dMZ}[3]{\DistChar_{#1}\pth{#2,#3}}%

\newcommand{\dGC}{\DistChar_\G}%
\newcommand{\dGY}[2]{\DistChar_\G\pth{#1,#2}}%
\newcommand{\dGZ}[3]{\DistChar_{#1}\pth{#2,#3}}%

\newcommand{\pp}{\Mh{p}}%
\newcommand{\pq}{\Mh{q}}%
\newcommand{\pz}{\Mh{z}}%

\newcommand{\sep}{\Mh{\mathsf{s}}}%

\newcommand{\GA}{\Mh{H}}%
\newcommand{\GP}{\Mh{\EuScript{G}}}

\newcommand{\etal}{\textit{et~al.}\xspace}

\renewcommand{\th}{th\xspace}

\newcommand{\pbrcx}[1]{\left[ {#1} \right]}%
\newcommand{\ProbLTR}{\Mh{\mathbb{P}}}%

\newcommand{\ProbC}{\ProbLTR}
\newcommand{\Prob}[1]{\mathop{\ProbLTR} \mleft[ #1 \mright]}%
\newcommand{\ProbCond}[2]{\mathop{\ProbLTR}\!\left[%
       #1 \;\middle\vert\; #2 \right]}

\newcommand{\ExChar}{\mathbb{E}}%
\newcommand{\ExSym}{\mathop{\ExChar}}%
\newcommand{\Ex}[2][\!]{\ExSym#1\pbrcx{#2}}

\newcommand{\NbrX}[1]{\Mh{\Gamma}\pth{#1}}%
\newcommand{\NbrY}[2]{\Mh{\Gamma}_{#1}\pth{#2}}%
\usepackage[inline]{enumitem}

\newlist{compactenumA}{enumerate}{5}%
\NotSubmitVer{
\setlist[compactenumA]{topsep=0pt,itemsep=-1ex,partopsep=1ex,parsep=1ex,%
   label=(\Alph*)}%
}
\SubmitVer{
\setlist[compactenumA]{label=(\Alph*)}
}

\newlist{compactenumI}{enumerate}{5}%
\NotSubmitVer{
\setlist[compactenumI]{topsep=0pt,itemsep=-1ex,partopsep=1ex,parsep=1ex,%
   label=(\Roman*)}%
}
\SubmitVer{
\setlist[compactenumI]{label=(\Roman*)}
}

\newlist{compactenumi}{enumerate}{5}%
\NotSubmitVer{
\setlist[compactenumi]{topsep=0pt,itemsep=-1ex,partopsep=1ex,parsep=1ex,%
   label=(\roman*)}%
}
\SubmitVer{
\setlist[compactenumi]{label=(\roman*)}
}

\newlist{compactitem}{itemize}{5}%
\setlist[compactitem]{topsep=0pt,itemsep=-1ex,partopsep=1ex,parsep=1ex,\label=ensuremath{\bullet}}%

\newlist{compactprop}{enumerate}{5}%
\setlist[compactprop]{topsep=3pt,left=10pt,itemsep=-0.3ex,partopsep=1ex,parsep=1ex,%
   label=(P\arabic*)}%

\newlength{\savedparindent}

\newcommand{\Term}[1]{\textsf{#1}}

\newcommand{\SpreadC}{\Mh{\Phi}}%
\newcommand{\SpreadX}[1]{\Mh{\SpreadC}\pth{#1}}%

\newcommand{\ballC}{\Mh{\mathsf{b}}}%
\newcommand{\ballY}[2]{\Mh{\mathsf{b}}\pth{#1, #2}}%
\providecommand{\remove}[1]{}%

\newcommand{\Here}{\typeout{LOCATION: \currfilename\space L\the\inputlineno}\xspace}

\newcommand{\Daniel}{D\'aniel\xspace}%
\newcommand{\Olah}{Ol\'ah\xspace}%

\providecommand{\Mh}[1]{{#1}}%

\newlength{\ppicX}
\newlength{\ppicY}

\newcommand{\Event}{\Mh{\mathcal{E}}}%
\newcommand{\EventB}{\mathcal{B}}%

\newcommand{\atgen}{\symbol{'100}} \newcommand{\SarielThanks}[1]{%
   \thanks{%
      Department of Computer Science; %
      University of Illinois; %
      201 N. Goodwin Avenue; Urbana, IL, 61801, USA; %
      \href{mailto:remove_this_sariel@illinois.edu}%
      {sariel@illnois.edu}; %
      {\tt \url{http://sarielhp.org/}.} #1%
   }%
}

\newcommand{\ManorThanks}[1]{%
   \thanks{%
      Department of Mathematics and Computer Science, The Open
      University of Israel, %
      {\tt manorme\atgen{}openu.ac.il}. %
      #1}%
}

\newcommand{\OlahThanks}[1]{%
   \thanks{%
      Department of Mathematics and Computing Science, TU Eindhoven,
      P.O. Box 513, 5600 MB Eindhoven, The Netherlands. %
      #1}%
}

\SubmitVer{
\theoremstyle{remark}
\newtheorem{remark}{Remark}
}

\newcommand{\SaveContent}[2]{%
   \expandafter\newcommand{#1}{#2}%
}

\newcommand{\RestatementOf}[2]{
   \noindent%
   \textbf{Restatement of #1.}
   {\em #2{}}%
}

\newcommand{\lossY}[2]{\Mh{\lambda}\pth{#1,#2}}

\allowdisplaybreaks[4]

\providecommand{\Erdos}{Erd{\H o}s\xspace}

\newcommand{\polylog}{\mathop{\mathrm{polylog}}}

\newcommand{\CSet}{\Mh{C}}%

\newcommand{\Metric}{\Mh{\mathcal{M}}}
\newcommand{\MS}[2]{\pth{#1,\Mh{{#2}}}}%

\newcommand{\FMS}{\EuScript{X}}

\newcommand{\dst}{\Mh{\xi}}%

\newcommand{\dmgY}[2]{\mathrm{dmg}\pth{#1,#2}}%

\newcommand{\cl}{\Mh{C}}%
\newcommand{\Depth}{\Mh{\mathsf{D}}}%
\newcommand{\Net}{\Mh{N}}%
\newcommand{\netY}[2]{\Mh{\mathrm{net}}\pth{#1,#2}}%
\newcommand{\BallSet}{\Mh{\mathcal{B}}}%

\newcommand{\lca}{\mathop{\mathrm{lca}}}
\newcommand{\HST}{\Term{HST}\xspace}

\newcommand{\parent}{\overline{\mathrm{p}}}
\newcommand{\parentX}[1]{\parent\pth{#1}}
\newcommand{\leqX}[1]{\stackrel{#1}\leq}
\newcommand{\Ball}{\Mh{\mathbb{B}}}%

\newcommand{\FlSet}{\Mh{\mathcal{F}}}%

\newcommand{\ManorHide}[1]{}

\newcommand{\Schaffer}{Sch{\"a}ffer\xspace}
\newcommand{\Szemeredi}{Szemer\'edi\xspace}

\providecommand{\TPDF}[2]{\texorpdfstring{#1}{#2}}

\begin{document}

\title{Reliable Spanners for Metric Spaces%
   \NotSubmitVer{%
      \thanks{%
         A preliminary version of this paper appeared in SoCG 2021
         \cite{hmo-rsms-21}.%
      }%
   }%
}%

\SubmitVer{%
\titlenote{%
   A preliminary version of this paper appeared in SoCG 2021
   \cite{hmo-rsms-21}.%
}%
}

\NotSubmitVer{%
   \author{%
      Sariel Har-Peled%
      \SarielThanks{Work on this paper was partially supported by a
         NSF AF award CCF-1907400.  }%
      \and%
      Manor Mendel%
      \ManorThanks{Supported by BSF Grant no. 2018223.}%
      \and%
      \Daniel \Olah%
      \OlahThanks{}%
   }%
}

\SubmitVer{%
   \author{Sariel Har-Peled}%
   \affiliation{%
      \institution{University of Illinois}%
      \streetaddress{201 N. Goodwin Avenue}%
      \city{Urbana}%
      \state{Illinois}%
      \country{USA} \postcode{61801}%
   }%
   \authornote{Work on this paper was partially supported by NSF AF
      award CCF-1907400.} %
   \email{sariel@illinois.edu}%
   \orcid{0003-2638-9635}
   %
   \author{Manor Mendel}%
   \affiliation{%
      \institution{
         The Open University of Israel}%
      \streetaddress{1 University Rd.; PO Box 808}%
      \city{Raanana}%
      \country{Israel}%
      \postcode{53537}%
   }%
   \authornote{Supported by BSF Grant no. 2018223.}
   \email{manorme@openu.ac.il}%
   \orcid{0002-7521-0358}
   \author{D\'a{n{}i{}e{}l} O{}l\'ah}%
   \affiliation{%
      \institution{
         TU Eindhoven}%
      \streetaddress{P.O. Box 513, 5600 MB}%
      \city{Eindhoven}%
      \country{The Netherlands}%
   }%
   \authornote{%
      Supported by the Netherlands Organisation for Scientific
      Research (NWO) through Gravitation-grant NETWORKS-024.002.003.
   }%
   \email{olahdani1025@gmail.com}%
   \orcid{0001-7179-3552}

   \renewcommand{\shortauthors}{Har-Peled, Mendel and \Olah}
}

\NotSubmitVer{\date{}}

\NotSubmitVer{\maketitle}%

\begin{abstract}
    A spanner is reliable if it can withstand large, catastrophic
    failures in the network. More precisely, any failure of some nodes
    can only cause a small damage in the remaining graph in terms of
    the dilation. In other words, the spanner property is maintained
    for almost all nodes in the residual graph. Constructions of
    reliable spanners of near linear size are known in the
    low-dimensional Euclidean settings. Here, we present new
    constructions of reliable spanners for planar graphs, trees and
    (general) metric spaces.
\end{abstract}

\SubmitVer{\maketitle}


\section{Introduction}

Let $\Metric=(\PS, \DistChar)$ be a finite metric space.  Let
$\G=(\PS,\Edges)$ be a sparse graph on the points of $\Metric$ whose
edges are weighted with the distances of their endpoints.  The graph
$\G$ is a \emphw{$t$-spanner} if for any pair of vertices
$u,v \in \PS$ we have $\dGY{u}{v} \leq t \cdot \dMY{u}{v}$, where
$\dGY{u}{v}$ is the length of the shortest path between $u$ and $v$ in
$\G$, and $\dMY{u}{v}$ is the distance in the metric space between $u$
and $v$.  Spanners were first introduced by Peleg and \Schaffer
\cite{ps-gs-89} as a tool in distributed computing, but have since
found use in many other areas, such as algorithms, networking, data
structures, and metric geometry, see~\cite{p-dclsa-00, ns-gsn-07}.

\paragraph*{Fault tolerant spanners}
A desired property of spanners is the ability to withstand failures of
some of their vertices. One such notion is provided by fault tolerance
\cite{clns-nds-15, lns-eacft-98, lns-iacft-02, l-nrftg-99,
   s-ofts-14}. A graph $\G$ is a \emphw{$k$-fault tolerant}
$t$-spanner, if for any subset of vertices $\BSet$, with
$\cardin{\BSet} \leq k$, the graph $\G \setminus \BSet$ is a
$t$-spanner.  However, for $k$-fault tolerant graphs there is no
guarantee if more than $k$ vertices fail, and furthermore, the size of
a fault tolerant graph grows (linearly) with the parameter $k$. In
particular, for fixed $t\geq 1$, the optimal size of $k$-fault
tolerant $(2t-1)$-spanners on $n$ vertices is
$\Of(k^{1-1/t} n^{1+1/t})$ \cite{bdpv-ovfts-18}. Note that vertex
degrees must be at least $\Omega(k)$ to avoid the possibility of
isolating a vertex.  Thus, it is not suitable for massive failures in
the network.

\paragraph*{Reliable spanners}
An alternative notion was introduced by Bose \etal \cite{bdms-rgs-13}.
Here, for a parameter $\epsR > 0$, a \emphw{$\epsR$-reliable
   $t$-spanner} has the property that for any failure (or attack) set
$\BSet$, the residual graph $\G \setminus \BSet$ is a $t$-spanner path
for the points of $\VV \setminus \EBSet$, where
$\EBSet\supseteq \BSet$ is some set, such that
$\cardin{\EBSet} \leq (1+\epsR)\cardin{\BSet}$.  We consider two
variants:
\begin{compactenumA}
    \NotSubmitVer{\smallskip}%
    \item \textbf{Adaptive adversary} (i.e., standard or
    ``deterministic'' model): The adversary knows the spanner $\G$,
    and the set $\BSet$ is chosen as a worst case for $\G$.

    \NotSubmitVer{\smallskip}%
    \item \textbf{Oblivious adversary} (i.e., ``randomized'' model):
    Here, the spanner $\G$ is drawn from a probability distribution
    $\chi$ (over the same number of vertices).  The adversary knows
    $\chi$ in advance, but not the sampled spanner. In this oblivious
    model, we require that $\Ex{|\EBSet|}\le (1+\epsR)|\BSet|$.
\end{compactenumA}
\NotSubmitVer{\smallskip}%
See \secref{prelims} for precise definitions.

\paragraph{Previous work}
 
Bose \etal \cite{bdms-rgs-13} provided some lower bounds and
constructions in the general settings, but the bounds on the size of
the damaged set (i.e., $\EBSet$) are much larger.  In the Euclidean
settings, for any point set $\PS \subseteq \Re^d$, and for any
constants $\epsR, \eps \in (0,1)$, one can construct a
$\epsR$-reliable $(1+\eps)$-spanner with only
$\Of\bigl(n \log n \log\!\log^6 n\bigr)$ edges \cite{bho-spda-19} (an
alternative, but slightly inferior construction, was provided
independently by Bose~\etal~\cite{bcdm-norgm-18}).  The number of
edges can be further reduced in the oblivious adversary case, where
one can construct an oblivious $(1+\eps)$-spanner that is
$\epsR$-reliable in expectation and has
$\Of( n \log\!\log^2n \log\!\log\!\log n )$ edges \cite{bho-srsal-20}.

\paragraph*{Covers.}

A cover is a set of clusters (i.e., subsets of the point set) that
covers the metric space with certain desirable properties.  Awerbuch
and Peleg \cite{ap-sp-90} showed a cover, where (i) each cluster has
diameter $O(k)$, (ii) every vertex participates in $O(k n^{1/k})$
clusters, and (iii) for every vertex $v$, the ball of radius $k$
centered at $v$, is contained in a single cluster.
Busch \etal \cite{blt-scpg-14} show how to compute a cover of a planar
graph, with diameter $\leq 16k$ per cluster, such that each pair in
distance $k$ from each other belongs to some cluster, and every vertex
participates in at most $18$ clusters. For graphs that excludes a
minor of fixed size, they get a similar result, except that each
vertex might participate in $O( \log n)$ vertices, where $n$ is the
number of vertices in the input graph.
Abraham \etal \cite{agmw-sddmf-10} presented a result with better
sparsity when the graphs do not have $K_{r,r}$ as a minor.

\begin{table}[t]
    \centering {%
       \begin{tabular}{|l*{4}{|c}|}
         \hline
         \multicolumn{4}{c}{Constructions of $\epsR$-reliable
         $\Delta$-spanners}\\
         \hline
         metric
         & $\Delta$
         & guarantee
         & size
         & ref\\
         \hline
         \hline
         \multirow{3}{*}{Uniform} 
         & $2$
         & expectation
         & $\Of(n \epsR^{-1} \log\epsR^{-1}) \Bigr.$ 
         & \lemref{reliable-star}%
         \\
         & $t$
         & {d}e{t}. lower bound
         & $\Omega( n^{1+{1}/{t}} )$
         & \lemref{u:lower:bound}%
         \\
         & $2t-1$
         & deterministic 
         & $\Of\pth{\epsR^{-2}n^{1+1/t}}\Bigr.$
         & \thmref{2t-2t-1-spanner}%
         \\
         & $2t$
         & deterministic 
         & $\Of\pth{\epsR^{-1}n^{1+1/t}}\Bigr.$
         & \thmref{2t-2t-1-spanner}%
         \\
         \hline%
         \multirow{4}{*}{Finite metrics} 
         & $\Of( \log n)$
         & expectation
         & $\Of\pth{ \epsR^{-1} n \polylog }$
         & \lemref{r:s:metric:rand}%
         \\
         & $\Of( t)$
         & expectation
         & $\Of\pth{ \epsR^{-1} n^{1+1/t} \polylog }$
         & \lemref{r:s:metric:rand}%
         \\
         & $\Of( t\log n)$
         & deterministic
         & $\Of\pth{\epsR^{-1}  n^{1+1/t}  \polylog }$
         & \lemref{r:s:metric:det}%
         \\
         & $\Of( t^2)$
         & deterministic
         & $\Of\pth{\epsR^{-1}  n^{1+1/t}  \polylog }$
         & \lemref{r:s:metric:det}%
         \\
         \hline%
         \multirow{2}{*}{Ultrametrics}
         & $2+\eps$
         & expectation
         & $\Of\pth{  \epsR^{-1} \eps^{-2} n \polylog }\Bigr.$
         & \lemref{r:s:ultra:rand}%
         \\
         & $(2+\eps)t-1$
         & deterministic
         & $\Of\pth{  \epsR^{-2} \eps^{-3}  n^{1+1/t}  \polylog }\Bigr.$
         & \lemref{r:s:ultra:det}%
         \\
         \hline
         \multirow{2}{*}{Trees}
         & $3+\eps$
         & expectation
         & $\Of\pth{  \epsR^{-1} \eps^{-2} n \polylog }\Bigr.$
         & \lemref{r:s:tree:rand}%
         \\
         & $(4+\eps)t-3$
         & deterministic
         & $\Of\pth{  \epsR^{-2} \eps^{-3}  n^{1+1/t}  \polylog }\Bigr.$
         & \lemref{r:s:tree:det}%
         \\
         \hline
         \multirow{2}{*}{Planar graphs}
         & $3+\eps$
         & expectation
         & $\Of\pth{   \epsR^{-1} \eps^{-4} n \polylog }\Bigr.$
         & \lemref{r:s:planar:rand}%
         \\
         & $(4+\eps)t -3$
         & deterministic
         & $\Of(\epsR^{-2} \eps^{-6}  n^{1+1/t}   \polylog   )$
         & \lemref{r:s:planar:det}%
         \\
         \hline
       \end{tabular}%
    }%
    \caption{Our results. Polylog factors are polynomial factors in
       $\log n$ and $\log \SpreadC$, where $\SpreadC$ is the spread of
       the metric. For trees and planar graphs, these results are for
       graphs with weights on the edges. Here in expectation denotes
       that the spanner works against an oblivious adversary (here,
       the expectation is over the randomization in the construction),
       and the guarantee is on the expected size of the damaged
       set. Similarly, deterministic implies an adaptive adversary.  }
    \tbllab{results}
\end{table}

\section*{Our results.}

We provide new constructions of reliable spanners for finite uniform
metrics, ultrametrics, trees, planar graphs and finite metrics.  Our
new results are summarized in \tblref{results}.

\paragraph*{Technique} Our approach for constructing reliable spanners
is in two steps: We first construct reliable spanners for uniform
metrics and then reduce the problem of constructing reliable spanners
for general metrics to uniform metrics using \emphi{covers}.

\paragraph*{Spanners for uniform metrics}
Uniform metrics have trivial classical 2-spanners -- that is, star
graphs.  In the oblivious model one can simply use ``constellation of
stars'' with a constant number of random centers, and the resulting
spanner is linear in size.  In the adaptive settings, we present a
lower bound of $\Omega(n^{1+1/t})$ edges for a reliable $t$-spanner,
and asymptotically ``matching'' construction of a deterministic
$(2t-1)$-reliable spanner with $\Of(n^{1+1/t})$ edges. The
construction is based on reliable expanders -- these are expanders
that remain expanding under the type of attacks described above.  See
\secref{uniform} and \secref{expander:spanner} for details.

\paragraph*{Covers}
A \emphi{$t$-cover} of a finite metric space $\Metric = (\PS,\distC)$
is a family of subsets $\Cover = \Set{ \SSet}{\SSet \subseteq \PS}$,
such that for each pair $\pp,\pq \in \PS$ of points there exists a
subset in $\Cover$ that contains both points and whose diameter is at
most $t\cdot\dMY{\pp}{\pq}$.  Covers are used here to extend reliable
spanners for uniform metrics into reliable spanners for general
metrics.  This is done by using spanners for uniform in each set of
the cover and then taking a union of the edges of those graphs.  See
\secref{covers:to:spanners}.

Naturally, the size $\sum_{\SSet \in \Cover} \cardin{\SSet}$ of a
$t$-cover $\Cover$ is an important parameter in the resulting size of
the spanner, so in \secref{covers} we study the problem of
constructing good covers.  For general $n$-point spaces with spread at
most $\SpreadC$, we observe that the Ramsey partitions
of~\cite{mn-rppds-07} provide $\Of(t)$ covers of size
$\Of(n^{1+1/t} \log \SpreadC)$, which is close to optimal, because of
an $\Omega(n(n^{1/t}+\log_t\SpreadC))$ lower bound we provide.  In
more specific cases, like ultrametrics, trees and planar graphs, one
can do better.  For example, for trees and planar graphs one gets
$(2+\eps)$-covers of near linear size.  For planar graphs, known
partitions have much larger gap, which makes these results quite
interesting.

\paragraph*{New reliable spanners}
Plugging the constructions of spanners for uniform metrics with the
construction of covers yields reliable spanners for finite uniform,
ultrametric, tree, planar, and general metrics. The results are
summarized in \tblref{results}.

\paragraph*{Efficient construction}

All our constructions relies on randomized constructions of expanders
(over $m$ vertices), that succeeds with probability
$\geq 1-1/m^{O(1)}$. As such, the constructions described can be done
efficiently, if one wants constructions of spanners for $n$ vertices
that succeeds with probability $1 - 1/n^{O(1)}$. See
\remref{constructive} for details.


\section{Preliminaries}
\seclab{prelims}


\subsection{Metric spaces}
For a set $\FMS$, a function
$\DistChar:\FMS^2 \rightarrow [0,\infty)$, is a \emphi{metric} if it
is symmetric, complies with the triangle inequality, and
$\dMY{\pp}{\pq}=0$ $\iff$ $\pp= \pq$.  A \emphi{metric space} is a
    pair $\Metric = (\FMS, \DistChar)$, where $\DistChar$ is a metric.
    For a point $\pp \in \FMS$, and a radius $r$, the \emphi{ball} of
    radius $r$ is the set
    \begin{equation*}
        \ballY{\pp}{r} = \Set{\pq \in \FMS}{\dMY{\pp}{\pq} \leq r}.    
    \end{equation*}
    For a finite set $X \subseteq \FMS$, the \emphi{diameter} of $X$
    is
    \begin{equation*}
        \diamX{X} = \diamY{\Metric}{X} = \max_{\pp,\pq \in X} \dMY{\pp}{\pq},    
    \end{equation*}
    and the \emphi{spread} of $X$ is
    \begin{math}
        \SpreadX{X} = \frac{\diamX{X}}{\min_{\pp,\pq \in X,
              \pp\neq\pq} \dMY{\pp}{\pq}}.
    \end{math}
    A metric space $\Metric = (\FMS, \DistChar)$ is \emphi{finite}, if
    $\FMS$ is a finite set.  In this case, we use
    $\SpreadC = \SpreadX{\FMS}$ to denote the spread of the (finite)
    metric.

    A natural way to define a metric space is to consider an
    undirected connected graph $\G = (\PS, \Edges)$ with positive
    weights on the edges. The \emphw{shortest path metric} of $\G$,
    denoted by $\dGC$, assigns for any two points $\pp,\pq \in \PS$
    the length of the shortest path between $\pp$ and $\pq$ in the
    graph. Thus, any graph $\G$ readily induces the finite metric
    space $(\VX{\G}, \dGC)$. If the graph is unweighted, then all the
    edges have weight $1$.

    A \emphi{tree metric} is a shortest path metric defined over a
    graph that is a tree.

    \subsection{Reliable spanners}
    \begin{definition}
        For a metric space $\Metric = (\PS, \DistChar) $, a graph
        $\GA = (\PS, \Edges)$ is a \emphi{$t$-spanner}, if for any
        $\pp,\pq \in \PS$,
        \begin{math}
            \dMY{\pp}{\pq}\leq \dMZ{\GA}{\pp}{\pq} \leq t \cdot
            \dMY{\pp}{\pq}.
        \end{math}
        Here $\DistChar_\GA$ is the shortest path distance on $H$
        whose edges are weighted according to $\DistChar$.
    \end{definition}

    Given a weighted graph $\G = (\VV, \Edges)$, and a set
    $\BSet\subseteq \VV$, we denote by $\G|_\BSet$ the subgraph
    induced on $\BSet$. We also use the notation
    $\G\setminus \BSet=\G|_{\VV\setminus \BSet}$.  A randomized graph
    $\G$ is a probability distribution over \emph{the edge set}
    $\Edges$ for a given set of vertices $\VV$.

    An \emphi{attack} on a graph $\G= (\VV, \Edges)$ is a set of
    vertices $\BSet$ that fails and no longer can be used.  An attack
    (on a randomized graph) is \emphi{oblivious}, if the set $\BSet$
    is picked stochastically independent of the edge set of the graph.

\begin{definition}[Reliable spanner] 
    \deflab{reliable:spanner}%
    Let $\G=(\PS,\Edges)$ be a $t$-spanner for some $t\geq 1$
    constructed by a (possibly) randomized algorithm.  Given an attack
    $\BSet$, its \emphi{damaged set} $\EBSet$ is a set of smallest
    possible size, such that for any pair of vertices
    $\pp,\pq \in \PS \setminus \EBSet$, we have
    \begin{equation} \eqlab{reliable:upper:bound} \dGZ{\G \setminus
           \BSet}{\pp}{\pq} \leq t \cdot \dMY{\pp}{\pq} ,
    \end{equation}
    that is, distances are preserved (up to a factor of $t$) for all
    pairs of points not contained in $\EBSet$.  The quantity
    $\cardin{\EBSet \setminus \BSet}$ is the \emphi{loss} of $\G$
    under the attack $\BSet$. %
    The \emphi{loss rate} of $\G$ is
    $\lossY{\G}{\BSet} = \cardin{\EBSet \setminus \BSet} /
    \cardin{\BSet}$. For $\epsR\in(0,1)$, the graph $\G$ is
    \emphi{$\epsR$-reliable} (in the deterministic or non-oblivious
    setting), if $\lossY{\G}{\BSet} \leq \epsR$ holds for any attack
    $\BSet \subseteq \PS$.  Furthermore, the graph $\G$ is
    \emphi{$\epsR$-reliable in expectation} (or in the oblivious
    model), if $\Ex{\lossY{\G}{\BSet}} \leq \epsR$ holds for any
    oblivious attack $\BSet \subseteq \PS$.
\end{definition}

\begin{remark}
    The damaged set $\EBSet$ is not necessarily unique, since there
    might be freedom in choosing the point to include in $\EBSet$ for
    a pair that does not have a $t$-path in $\G\setminus \BSet$.
\end{remark}

\subsection{Miscellaneous}

For a graph $\G$, and two set of vertices $Y,Z \subseteq \VX{\G}$, let
\begin{equation*}
    \NbrY{Z}{Y} = \Set{x \in Z}{ xy \in \EdgesX{\G} \text{ and } y \in
       Y}     
\end{equation*}
denote the \emphi{neighbors} of $Y$ in $Z$. The neighbors of $Y$ in
$\G$ is denoted by $\NbrX{Y}=\NbrY{\VX{\G}}{Y}$.

\begin{definition}
    \deflab{depth}%
    For a collection of sets $\Family$, and an element $x$, let
    $\degY{x}{\Family} = \cardin{\Set{X \in \Family}{ x \in X}}$
    denote the \emphi{degree} of $x$ in $\Family$. The maximum degree
    of any element of $\Family$ is the \emphi{depth} of $\Family$.
\end{definition}

\paragraph*{Notations}
We use $\PS + \pp = \PS \cup \{ \pp\}$ and
$\PS - \pp = \PS \setminus \{ \pp\}$. Similarly, for a graph $\G$, and
a vertex $\pp$, let $\G - \pp$ denote the graph resulting from
removing $\pp$.


\section{Reliable spanners for uniform metric}
\seclab{uniform}%

Let $\PS$ be a set of $n$ points and let $\MS{\PS}{\DistChar}$ be a
metric space equipped with the uniform metric, that is, for all
distinct pairs $\pp,\pq \in \PS$, we have that $\dMY{\pp}{\pq}$ is the
same quantity (e.g., $1$).  Note that $n-1$ edges are enough to
achieve a $2$-spanner for the uniform metric by using the star graph.


\subsection{A randomized construction for the oblivious %
   case}
\seclab{unif-rel-star}

\paragraph*{Construction}
Let $\epsR \in (0,1)$ be a fixed parameter. Set
\begin{math}
    k=2\ceil{\epsR^{-1} \log\epsR^{-1}} +1
\end{math}
and sample $k$ points from $\PS$ uniformly at random (with
replacement). Let $\CSet \subseteq \PS$ be the resulting set of
\emphw{center points}. For each point $\pp \in \CSet$, build the star
graph $\StarX{\pp} = (\PS, \Set{ \pp\pq}{\pq \in \PS - \pp}) $, where
$\pp$ is the center of the star. The \emphi{constellation} of $\CSet$
is the graph $\StarR = \bigcup_{\pp\in \CSet} \StarX{\pp}$, which is
the union of the star graphs induced by centers in $\CSet$.

\begin{lemma}
    \lemlab{reliable-star}%
    The constellation $\StarR$, defined above, is a $\epsR$-reliable
    $2$-spanner in expectation. The number of its edges is
    $\Of(n \epsR^{-1} \log\epsR^{-1})$.
\end{lemma}

\begin{proof}
    Let $\BSet\subseteq \PS$ be an oblivious attack, and let
    $b=\cardin{\BSet}$. If there is a point of $\CSet$ that is not in
    $\BSet$, then there is a center point outside of the attack set,
    which provides $2$-hop paths between all pairs of points in the
    residual graph, and therefore we choose $\EBSet=\BSet$.  On the
    other hand, if $\CSet \subseteq \BSet$, then the residual graph
    contains only isolated vertices, and therefore we choose
    $\EBSet=\PS$.  If $(1+\epsR)b \geq n$ then there is nothing to
    prove. Thus, since $b/n < 1/(1+\epsR)$ and
    $1/( 1+\epsR) \leq 1-\epsR/2$, we have
    \begin{align*}
      \Ex{\lossY{\G}{\BSet}} 
      &=
        0  \Prob{\CSet \nsubseteq  \BSet}
        +%
        \frac{n-1-b}{b} \Prob{\CSet \subseteq \BSet}
        \leq
        \frac{n}{b} \pth{\frac{b}{n}}^k
        \leq %
        \frac{1}{ ( 1+\epsR)^{k-1}}
        \leq %
        \exp\Bigl( -\frac{k-1}{2} \epsR \Bigr)
        \leq
        \epsR.
    \end{align*}

    As for the number of edges, $\StarR$ has at most $k(n-1)$ edges,
    since $\StarR$ is the union of $k$ stars and each star has $n-1$
    edges. Thus, by the choice of $k$, the size of $\StarR$ is
    $\Of(n \epsR^{-1} \log\epsR^{-1})$.
\end{proof}


\subsection{Lower bound for a deterministic construction}
\seclab{unif-lower}

In the \emphw{non-oblivious settings}, the attacker knows the
constructed graph $\G$ when choosing the attack set $\BSet$.

\begin{lemma}
    \lemlab{u:lower:bound}%
    Let $\G=(P,E)$ be a $\epsR$-reliable $t$-spanner on $\PS$ for the
    uniform metric, where $\epsR\in (0,1)$ and $t\geq 1$. Then, in the
    non-oblivious settings, $\G$ must have $\Omega( n^{1+{1}/{t}} )$
    edges.
\end{lemma}
\begin{proof}
    We assume that the distance between any pair of points of $P$ is
    one.  Let the attack set $\BSet$ be the set of all nodes of degree
    at least $\Delta$, where $\Delta = n^{1/t}/4$.  Assume, toward a
    contradiction, that $\cardin{\BSet} \leq n/4$.  By the reliability
    condition, there exists a set $\QS \subseteq \PS\setminus \BSet$
    of nodes of the residual graph of size at least
    $n-(1+\epsR)\cardin{\BSet} \geq n/2$, that has $t$-hop paths in
    $\G\setminus \BSet$ between all pairs of vertices of $\QS$. Let
    $\pp \in \QS$ be an arbitrary vertex.  Let $\ballY{\pp}{t}$ be a
    ball of radius $t$ centered at $\pp$ in the shortest path metric
    in the residual graph $\G \setminus \BSet$.  On the one hand, from
    the above, $\ballY{\pp}{t}\supseteq Q$.  On the other hand, the
    maximum degree is at most $\Delta$, and therefore
    $\cardin{\ballY{\pp}{t}} \leq \Delta^t \leq n/4 $.  As such, we
    have
    $n/4 \geq \cardin{\ballY{\pp}{t}} \geq \cardin{\QS} \geq n/2$. A
    contradiction.
    
    Hence, $\cardin{\BSet} >n/4$.  The claim follows since
    $\Delta |\BSet| \leq 2 |\Edges|$, which implies
    \begin{math}
        \cardin{\Edges} \geq \Delta \cardin{\BSet}/2 = \Omega(
        n^{1+1/t}).
    \end{math}
\end{proof}

\begin{remark}
    \remlab{erdos}%
    \Erdos' girth conjecture states that there exists a graph $\G$
    with $n$ vertices and $\Omega(n^{1+1/k})$ edges, and girth at
    least $2k+1$, where the \emphw{girth} of $\G$ is the length of the
    shortest cycle in $\G$.  The argument in the proof of
    \lemref{u:lower:bound} is closely related to the standard argument
    for proving a tight counterpart -- any graph with
    $\Omega(n^{1+1/k})$ edges has girth at most $2k+1$.
\end{remark}


\subsection{Reliable %
   spanners of the uniform metric for adaptive adversary}
\seclab{unif-upper}

Here, we present a construction of reliable spanner that is close to
being tight. The spanner is simply a high-degree expander whose
properties are described in the following definition.

\newcommand{\M}{\Mh{M}}

%

\begin{definition}
    \deflab{proper:expander}%
    Denote by $\lambda(\G)$ the second eigenvalue of the matrix
    $\M/d$, where $\M = \mathrm{Adj}(\G)$ is the adjacency matrix of a
    $d$-regular graph $\G$.  A \emphi{proper expander} specifies a
    constant $C>1$, and functions $\CfX{\delta},\cfX{\delta} > 0$,
    such that for every $\delta\in(0,1/4)$ and even integers
    $d\geq \CfX{\delta} $, $n\geq d^2$, there exists an $n$-vertex,
    $d$-regular graph $\G=(\VV,\Edges)$, such that:
    \begin{compactprop}
        \item \itemlab{p:1}
        \begin{math}
            \forall S\subseteq \VV,\ |S|%
            \geq%
            {12n}/(\delta d)%
            \; \implies \; |\NbrX{S}| > (1-\delta)n,
        \end{math}

        \item \itemlab{p:2}
        \begin{math}
            \forall S\subseteq \VV,\ |S| \leq%
            \cfX{\delta} n/d \; \implies\; |\NbrX{S}| \geq (1-\delta)
            d |S|.
        \end{math}

        \item 
        \itemlab{p:3}
        \begin{math}
            \lambda(\G) \leq {C}/{\sqrt d}.
        \end{math}

    \end{compactprop}
\end{definition}

For each one of the properties above, it is known that there exists an
expander satisfying it: Property~\itemref{p:1} is essentially proved
in~\cite{bho-spda-19}, Property~\itemref{p:2} is folklore, and
Property~\itemref{p:3} appears in~\cite{djpp-fltrr-13}.  Since they
hold almost surely for ``random regular graphs'', they also hold
simultaneously.  However, we were unable to find in the literature
proofs of almost sure existence in the same model of random regular
graphs, and contiguity of the different random models does not
necessarily hold in the high-degree regime (which is what we need
here).  Therefore, for completeness, \apndref{proper:expander} reprove
\itemref{p:1} and \itemref{p:2} in the same random model in which
\itemref{p:3} was proved. We thus get the following:

\SaveContent{\ThmProperExapdnerBody}%
{%
   The random graph constructed in \secref{construction} is a proper
   expander (see \defref{proper:expander}), asymptotically almost
   surely. Specifically, the probability the graph has the desired
   properties is $\geq 1- n^{-O(1)}$.%
}%
\begin{theorem}
    \thmlab{proper:expander}%
    \ThmProperExapdnerBody{}
\end{theorem}

With the appropriate choice of parameters, these expanders are
reliable spanners for uniform metrics.

\begin{theorem}
    \thmlab{2t-2t-1-spanner}%
    For every $t\in \NN$, $\theta\in(0,1)$, and $n\in 2\NN$, such that
    $n\ge e^{\Omega(t)}$, there exist:
    \begin{compactenumi}
        \NotSubmitVer{\smallskip}
        \item $\epsR$-reliable $2t$-spanner with
        $\Of{(\epsR^{-1} n^{1+1/t})}$ edges for $n$-point uniform
        space, and

        \NotSubmitVer{\smallskip}
        \item $\epsR$-reliable $(2t-1)$-spanner with
        $\Of{(\epsR^{-2} n^{1+1/t})}$ edges for $n$-point uniform
        space.
    \end{compactenumi}
\end{theorem}

The proof is somewhat cumbersome and is deferred to
\secref{expander:spanner}.

Applying the above theorem directly on non-uniform metric we obtain
the following corollaries.

\begin{corollary}
    \corlab{det:uniform:r:s:2t-1}%
    Let $\Metric = (\FMS, \DistChar)$ be a metric space, and let
    $\PS \subseteq \FMS$ be a finite subset of size $n$.  Given
    parameters $t \in \NN$, and $\epsR \in (0,1)$, there exists a
    weighted graph $\G$ on $\PS$, such that:
    \begin{compactenumA}
        \NotSubmitVer{\smallskip}%
        \item The graph $\G$ has
        $\cardin{\EdgesX{\G}} = \Of\pth{\epsR^{-2}\cdot n^{1+1/t}}$
        edges.
        
        \NotSubmitVer{\smallskip}%
        \item The graph $\G$ is $\epsR$-reliable. Namely, given any
        attack set $\BSet\subset \FMS$, there exists a subset
        $\QS \subseteq \PS$, such that
        $\cardin{\QS} \geq \cardin{\PS} - (1+\epsR)\cardin{\BSet \cap
           \PS}$. Furthermore, for any two points $\pp,\pq \in \QS$,
        we have
        \begin{equation*}
            \dMZ{\Metric}{\pp}{\pq} \leq \dMZ{\G|_\QS}{\pp}{\pq} \leq
            (2t-1)\cdot \diamY{\Metric}{\PS},
        \end{equation*}
        and the path realizing it has at most $(2t-1)$ hops.
    \end{compactenumA}
    \NotSubmitVer{\smallskip}%
    In particular, $\G$ has hop diameter at most $2t-1$, and diameter
    at most $(2t-1) \cdot \diamY{\Metric}{\PS}$.
\end{corollary}

\begin{corollary}
    \corlab{det:uniform:r:s:2t}%
    Let $\Metric = (\FMS, \DistChar)$ be a metric space, and let
    $\PS \subseteq \FMS$ be a finite subset of size $n$.  Given
    parameters $t \in \NN$, and $\epsR \in (0,1)$, there exists a
    weighted graph $\G$ on $\PS$, such that:
    \begin{compactenumA}
        \NotSubmitVer{\smallskip}%
        \item The graph $\G$ has
        $\cardin{\EdgesX{\G}} = \Of\pth{\epsR^{-1}\cdot n^{1+1/t}}$
        edges.
        
        \NotSubmitVer{\smallskip}%
        \item The graph $\G$ is $\epsR$-reliable. Namely, given any
        attack set $\BSet\subset \FMS$, there exists a subset
        $\QS \subseteq \PS$, such that
        $\cardin{\QS} \geq \cardin{\PS} - (1+\epsR)\cardin{\BSet \cap
           \PS}$. Furthermore, for any two points $\pp,\pq \in \QS$,
        we have
        \begin{equation*}
            \dMZ{\Metric}{\pp}{\pq} \leq \dMZ{\G|_\QS}{\pp}{\pq} \leq
            2t\cdot \diamY{\Metric}{\PS},
        \end{equation*}
        and the path realizing it has at most $2t$ hops.
    \end{compactenumA}
    \NotSubmitVer{\smallskip}%
    In particular, $\G$ has hop diameter at most $2t$, and diameter at
    most $2t \cdot \diamY{\Metric}{\PS}$.
\end{corollary}

\begin{remark}
    \corref{det:uniform:r:s:2t-1} and \corref{det:uniform:r:s:2t} are
    quite weak as far as the guarantee on the length of the resulting
    path (i.e., they are not good spanners).  A construction that
    provides a spanner guarantee is provided below in
    \lemref{r:s:metric:det} below.
\end{remark}

\SaveContent{\ThmRelVertexExpanderBody}%
{%
   For every $\epsR\in(0,1)$ there exists a constant
   $\cfX{}=\cfX{\epsR}>0$, such that for any expansion constant
   $h\ge \cfX{}^{-2}$, there exists a family of vertex expander graphs
   $\{\G=(\VV,\Edges)\}_n$ on $n$ vertices of degree at most
   $4h/\epsR$, with the following resiliency property: For any
   $\BSet\subset \VV$, there exists $\EBSet\supseteq \BSet$,
   $|\EBSet|\le (1+\epsR)|\BSet|$ such that the graph
   $\G\setminus \EBSet$ is a vertex expander in the sense that
   \begin{compactenumi}
       \NotSubmitVer{\smallskip}%
       \item $\diamX{\G\setminus \EBSet}\le 2\lceil \log_h n\rceil$,
       and \NotSubmitVer{\smallskip}%
       \item For any $U\subset \VV\setminus \EBSet$ of size
       $|U|\leq \cfX{} n /h$, we have
       $|\Gamma_{\G\setminus \EBSet}(U)|\geq h|U|$.
   \end{compactenumi}
}

As aside, proper expanders are reliable spanners because their
expansion property is robust, as testified by the following.
\begin{theorem}[reliable vertex expander]
    \thmlab{reliable:vertex:expander}%
    \ThmRelVertexExpanderBody{}
\end{theorem}

The proof is deferred to \secref{r:v:e:proof}.



\section{Covers for trees, bounded spread %
   metrics, and planar graphs}
\seclab{covers}

\begin{definition}
    For a finite metric space $\Metric = (\PS,\distC)$, a
    \emphi{$t$-cover}, is a family of subsets
    $\Cover = \Set{ \SSet_i \subseteq \PS}{i=1,\ldots m}$, such that
    for any $\pp,\pq \in \PS$, there exists an index $i$, such that
    $\pp,\pq \in \SSet_i$, and
    \begin{equation*}
        \diamX{ \SSet_i}/t \leq  \dMY{\pp}{\pq} \leq \diamX{ \SSet_i}.
    \end{equation*}
    The \emphi{size} of a cover $\Cover$ is
    $\mathrm{size}(\Cover)=\sum_{\SSet \in \Cover} \cardin{\SSet}$.
    Recalling \defref{depth}, the \emphi{degree} in $\Cover$ of a
    point $\pp \in \PS$ is the number of sets of $\Cover$ that contain
    it.  The \emphi{depth} of $\Cover$ is the maximum degree of any
    element of $\PS$, and is denoted by $\Depth(\Cover)$.
\end{definition}

\subsection{Lower bounds}

Unfortunately, in the worst case, the depth and the size of any cover
must depend on the spread of the metric.

\begin{proposition}
    \proplab{cover:lower:bound}%
    For any parameter $t >1$, any integer $h> 1$, $\SpreadC = t^h$,
    and any $n\ge h$, there exists a metric $\Metric=(\PS,\DistChar)$
    of $n$ points, such that
    \begin{compactenumI}
        \item $\SpreadX{\PS} = \SpreadC$, and
    
        \item any $t$-cover $\Cover$ of $\PS$ has size
        $\Omega(n \log_t \SpreadC) = \Omega(nh)$, average degree
        $\geq h/2$, and depth $h$.
    \end{compactenumI}
\end{proposition}

\begin{proof}
    For simplicity of exposition, assume that $h$ divides $n$. Let
    $\PS_i$ be a set of $n/h$ points, such that the distance between
    any pair of points of $\PS_i$ is $\ell_i = (t+\eps)^i$, for some
    fixed small $\eps >0$, for $i=1,\ldots,h$. Furthermore, assume
    that the distance between any point of $\PS_i$ and any point of
    $\PS_j$, for $i < j$, is $\ell_j$.

    Let $\PS = \cup \PS_i$, and observe that the distance function
    defined above is a metric (it is the union of uniform metrics of
    different resolutions). Now, consider any $t$-cover $\Cover$ of
    $\PS$. The \emphw{rank} of a cluster $\cl \in \Cover$, is the
    highest $j$, such that $\PS_j \cap \cl \neq\emptyset$.  We can
    assume that all the clusters of $\cl$ have at least two points, as
    otherwise they can be removed. For any index $j \in \IRX{h}$, any
    point $\pp\in \PS_j$ and any point
    $\pq \in \cup_{i=1}^j \PS_i - \pp$, by definition, there exists a
    cluster $\cl \in \Cover$, such that $\pp,\pq \in \cl$, and
    $\diamX{\cl} \leq t \cdot \ell_j$, since
    $\dMY{\pp}{\pq} = \ell_j$. It follows that $\cl$ can not contain
    any point of $\cup_{i=j+1}^h \PS_i$.  Namely, the rank of $\cl$ is
    $j$.

    If there are two clusters of rank $j$ in $\Cover$, then we can
    merge them, since merging does not increase their diameter, and
    such an operation does not increase the size of the cover, and the
    degrees of its elements. As such, in the end of this process, the
    cover $\Cover$ has $h$ clusters, and for any $j \in \IRX{h}$,
    there is a cluster $\cl_j \in \Cover$ that is of rank $j$, and
    contains (exactly) all the elements of $\cup_{i=1}^j \PS_i$.
    Namely, $\Cover = \{ \cl_1, \ldots, \cl_h \}$, where
    $\cl_j = \cup_{i=1}^j \PS_i$. It is easy to verify that this cover
    has the desired properties.
\end{proof}

\begin{proposition}
    For any $t\in{2,3,\ldots}$, and any sufficiently large $n$ there
    exists an $n$-point metric space for which any $t$-cover must be
    of size at least $\Omega(n^{1+{1}/{2t}})$.
\end{proposition}
\begin{proof}
    Let $g=2t+2$.  By a standard probabilistic argument, for any
    sufficiently large $n\in \NN$ there exists a simple graph
    $G=(V,E)$ on $n$ vertices and $m=\Omega(n^{1+\frac{1}{g-2}})$
    edges whose girth is at least $g$. (This is not the best known
    bound, but it is sufficient for our purposes.)  Consider $G$ as a
    metric space with the shortest (unweighted) path metric.  Let
    $\Cover$ be a $t$-cover for $G$, and let
    $\Cover'=\Set{S\in \Cover }{ \diamX{S}\leq t}$.  For $S\subset V$
    let $E(S)=\Set{ uv \in E }{ u,v\in S}$.
    
    We claim that for every $S\in \Cover'$, $E(S)$ is a forest, and
    hence $|E(S)|<|S|$.  Indeed, Suppose $E(S)$ contains a cycle
    $(v_0,v_1,\ldots,v_h,v_0)$ such that
    $v_iv_{i+1}, v_0v_h\in E(S)\subseteq E$.  Denote
    $d_i=\dMY{v_0}{v_i}$. Since $\{v_i,v_{i+1}\}\in E$, we have
    $d_{i+1}-d_i\in\{-1,0,1\}$. Let $j$ be the smallest $i$ such that
    $d_{i+1}\le d_i$. Hence, $d_j=j$.  Let
    $P=(v_{j+1},u_1,\ldots u_k,v_0)$ be a shortest path in $G$ between
    $v_{j+1}$ and $v_0$. $P$'s length is at most $d_j=j$.  Thus, the
    sequence $v_0,v_1,\ldots v_j,v_{j+1},u_1,\ldots,u_k,v_0$ is closed
    path of length at most $2j+1$, and hence contains a cycle of
    length at most $2j+1$. By the girth condition, $2j+1\ge 2t+2$, and
    hence $d_j=j>t$. But since $v_0,v_j\in S$, this means that
    $\diamX{S}\ge d_j >t$ which contradicts the definition of
    $\Cover'$.

    For every edge $pq \in E$, $\dMY{p}{q}=1$, and by the $t$-covering
    condition there exists $S\in\mathcal C'$ such that $p,q\in S$.
    Therefore $pq\in E(S)$. Hence
    $E\subset \bigcup_{S\in \Cover'}E(S)$, and therefore
    \begin{equation*}
        c n^{1+\tfrac{1}{2t}}\le |E|\le \sum_{S\in \mathcal C'}|E(S)|
        < \sum_{S\in \Cover'}|S|=\mathrm{size}(\Cover')\le
        \mathrm{size}(\Cover).        
        \SubmitVer{\qedhere}%
    \end{equation*}
\end{proof}


\subsection{Cover for ultrametrics}

\begin{definition}
    \deflab{h:s:t}%
    A \emphic{hierarchically well-separated tree}{HST@\Term{HST}}
    (\emphOnly{\HST{}}) is a metric space defined on the leaves of a
    rooted tree $\Tree$. To each vertex $u \in \Tree$ there is an
    associated label $\Delta_u \ge 0$. This label is zero for all the
    leaves of $\Tree$, and it is a positive number for all the
    interior nodes.  The labels satisfy for every non-root vertex
    $v\in \Tree$, $\Delta_v \leq \Delta_{\parentX{v}}/k$, where
    $\parentX{v}$ is the parent of $v$ in $\Tree$.  The distance
    between two leaves $x,y\in \Tree$ is defined as
    $\Delta_{\lca(x,y)}$, where $\lca(x,y)$ is the least common
    ancestor of $x$ and $y$ in $\Tree$.  An \HST $\Tree$ is a
    \emphic{$k$-HST}{HST@\Term{HST}!k@$k$-\Term{HST}} if for every
    non-root vertex $v \in \Tree$,
    $\Delta_v \leq \Delta_{\parentX{v}}/k$.
\end{definition}

\HST{}s are one of the simplest non-trivial metrics, and surprisingly,
general metrics can be embedded randomly into HSTs with expected
distortion of $\Of( \log n )$ \cite{b-aamtm-98,frt-tbaam-04}.

\begin{definition}
    A metric $ \Metric=(\PS,\DistChar)$ is an \emphi{ultrametric}, if
    for any $x,y,z \in \PS$, we have that
    $\dMY{x}{z} \leq \max( \dMY{x}{y}, \dMY{y}{z} )$. Notice, that
    this is a stronger version of the triangle inequality, which
    states that $\dMY{x}{z} \leq \dMY{x}{y} + \dMY{y}{z} $.
\end{definition}
The following is folklore, and it also easy to verify (see, e.g.,
\cite[Lemma~3.5]{blmn-mrtp-05}.
\begin{lemma}
    For $k\ge 1$, every finite ultrametric can be $k$-approximated by
    a $k$-HST.
\end{lemma}

\begin{lemma}
    For $k>1$, every $k$-HST with spread $\SpreadC$ has $1$-cover of
    depth at most $\log_k \SpreadC$.
\end{lemma}
\begin{proof}
    Let $T$ be the tree of the HST. With every vertex $u\in T$ we
    associate a cluster $C_u$ the leaves of the subtree rooted at
    $u$. The covers is defined as $\Cover=\Set{C_u}{ u\in T}$.  The
    properties of the cover are immediate.
\end{proof}

\begin{corollary}
    \corlab{h:s:t:cover}%
    Let $\Metric = (\PS, \DistChar)$ be an ultrametric over $n$ points
    with spread $\SpreadC$. For any $\eps \in (0,1)$, one can compute
    a $(1+\eps)$-cover of $\Metric$ of depth
    $\Of( \eps^{-1} \log \SpreadC)$.
\end{corollary}

\subsection{Cover for general finite metrics}

We need the following result.
\begin{lemma}[\cite{mn-rppds-07}]
    \lemlab{m:n}%
    Let $(\PS, \DistChar)$ be an $n$-point metric space and
    $k \geq 1$. Then there exists a distribution over decreasing
    sequences of subsets
    $\PS = \PS_0 \supsetneq \PS_1 \supsetneq \cdots \supsetneq \PS_s =
    \emptyset$ ($s$ itself is a random variable), such that for all
    $\mu > -1/k$, we have
    \begin{math}
        \Ex{ \sum_{j=1}^{s-1} |\PS_j|^\mu } \leq%
        \max \pth{ \frac{k}{1 + \mu k}, 1} \cdot n^{\mu+1/k},
    \end{math}
    and such that for each $j \in \IRX{s}$ there exists an ultrametric
    $\rho_j$ on $\PS_{j-1}$ satisfying for every $\pp, \pq \in \PS$,
    that $\rho_j (\pp, \pq) \geq \dMY{\pp}{\pq}$, and if
    $\pp \in \PS_{j-1}$ and $\pq \in \PS_{j-1} \setminus \PS_j$ then
    $\rho_j (\pp, \pq) \leq \Of(k) \cdot \dMY{\pp}{\pq}$.
\end{lemma}

By computing a cover (using \corref{h:s:t:cover}) for each ultrametric
generated by the above lemma, we get the following.

\begin{lemma}
    \lemlab{cover:metric:spread:b}%
    For an $n$-point metric space $\Metric = (\PS,\distC)$ with spread
    $\SpreadC$, and a parameter $k >1$, one can compute, in polynomial
    time, an $\Of(k)$-cover of $\Metric$ of size
    $\Of( n^{1+1/k}\log \SpreadC)$ and depth
    $\Of( k n^{1/k} \log \SpreadC)$.
\end{lemma}
\begin{proof}
    Using \lemref{m:n}, for the parameter $k$, compute the sequence
    $\PS = \PS_0 \supsetneq \PS_1 \supsetneq \cdots \supsetneq \PS_s
    =\emptyset$, and the associated ultrametrics
    $\rho_1, \ldots, \rho_{s-1}$. We build a $1$-\HST $\Tree_i$ for
    $\rho_i$, when restricted to the set $\PS_{i-1}$, for
    $i=1,\ldots, s-1$. For every \HST in this collection, we compute a
    $2$-cover using \corref{h:s:t:cover}. Let $\Cover$ be the union of
    all these covers.

    Since for every $\HST$ $\Tree_i$ the resulting cover has size
    $\Of( \cardin{\PS_i} \log \SpreadC)$, then by \lemref{m:n}
    (applied with $\mu =1$), we have
    \begin{equation*}
        \Ex{\Bigl.\smash{\sum_i} \cardin{\PS_i} \log \SpreadC}%
        =%
        \Of(n^{1+1/k} \log \SpreadC).
    \end{equation*}

    As for the quality of the cover, let $\pp, \pq$ be any two points
    in $\PS$, and assume (without loss of generality) that
    $\pp \in \PS_{i-1}$ and $\pq \in \PS_{i-1} \setminus \PS_i$ for
    some $i\in[s]$.  We have that
    $\rho_i(\pp,\pq) = \Of( k ) \cdot \dMY{\pp}{\pq}$, and since we
    computed a $2$-cover for this tree, there is a cluster in the
    computed cover that contains both points and its diameter is at
    most twice the distance between those points.

    The maximum depth, follows by using \lemref{m:n} with $\mu
    =0$. This implies that a point of $\PS$ participates in
    $s=\Of( k n^{1/k})$ \HST{}s, and each cover generated by such an
    \HST might a point an element at most $\Of( \log \SpreadC )$
    times.

    The bounds on the size and depth hold in expectation, and one can
    repeat the construction if they exceed the desired size by (say) a
    factor of eight. In expectation, after a constant number of
    iterations, the algorithm would compute with high probability a
    cover with the desired bounds.
\end{proof}

\subsection{Covers for trees}

Using a tree separator, and applying it recursively, implies the
following construction of covers for trees.

\begin{lemma}
    \lemlab{cover:trees}%
    For a weighted tree metric $\Tree = (\PS,\distC)$, with spread
    $\SpreadC$, and a parameter $\eps \in (0,1)$, one can compute in
    polynomial time a $(2+\eps)$-cover of $\Tree$ of depth
    $\Of( \eps^{-1} \log \SpreadC \log n)$, and size
    $\Of( n \eps^{-1} \log \SpreadC \log n)$, where
    $n = \cardin{\PS}$.
\end{lemma}
\begin{proof}
    The proof is by induction on $n$. When $n=1$ the trivial cover is
    sufficient.  Assume next that $n\ge 1$ and the minimum distance in
    $\Tree$ is one.  Find a separator node $\sep$ in $\Tree$ such that
    there is no connected component in $\Tree$ larger than $n/2$ after
    removing $\sep$. For
    $i\in\{0,1,2,\dots,m = \ceil{\smash{\log_{1+\eps/2} \SpreadC}
    }\}$, define the sets
    \begin{math}
        \PS(\sep,i)%
        =%
        \Set{ \pp \in \PS}{ \dMY{\sep}{\pp} \leq (1+\eps/2)^i }.
    \end{math}
    By the inductive hypothesis, for each connected component $Q$ of
    $\Tree - \sep$ there is a $(2+\eps)$-cover $\Cover_Q$ of $Q$ of
    depth $\Of( \eps^{-1} \log \SpreadC \log (n/2))$ and size
    $\Of( n \eps^{-1} \log \SpreadC \log (n/2))$.  The cover for $P$
    is
    \begin{equation*}
        \Cover_P= \Set{\PS(\sep,i)}{ i=0,1,2,\dots,m }%
        \cup%
        {\bigcup\nolimits_Q} \Cover_Q.        
    \end{equation*}
    Since every element of $\PS$ participates in at most $m$ sets in
    each level of the recursion, the bound on the depth and size is
    immediate.

    As for the quality of the cover, by the inductive hypothesis, we
    need only to check pairs of points $\pp,\pq \in \PS$, that are
    separated by $\sep$. Assume
    $\dMY{\sep}{\pp} \geq \dMY{\sep}{\pq}$, and let $j$ be the minimum
    index, such that $\dMY{\sep}{\pp} \leq (1+\eps/2)^j$. We have that
    $\pp,\pq \in \PS(\sep,j)$, and
    \begin{equation*}
        \dMY{\pp}{\pq}%
        \geq%
        \dMY{\sep}{\pp}%
        \geq%
        (1+\eps/2)^{j-1}
        \qquad\text{and}\qquad%
        \diamX{\PS(\sep,j)}\leq 2 (1+\eps/2)^j.        
    \end{equation*}
    As such, we have
    \begin{math}
        \displaystyle%
        \frac{\diamX{\PS(\sep,j)}}{\dMY{\pp}{\pq}} \leq %
        \frac{2 (1+\eps/2)^j}{(1+\eps/2)^{j-1}}%
        =%
        2+\eps.
    \end{math}
\end{proof}

\subsection{Covers for planar graphs}

The next lemma can be traced back to the work of Lipton and Tarjan
\cite{lt-stpg-79}.

\begin{lemma}
    \lemlab{separator}%
    Let $\GA=(\PS,\Edges)$ be a planar triangulated graph with
    non-negative edge weights. There is a partition of $\PS$ to three
    sets $\PA,\PB,\PC$, such that \begin{compactenumi}
        \item $\cardin{\PA} \leq (2/3)n$ and
        $\cardin{\PB} \leq (2/3)n$,
        \item there is no edge between $\PA$ and $\PB$, and
        \item $\PC$ is composed of two interior disjoint shortest
        paths that share one of their endpoints, and an edge
        connecting their other two endpoints.
    \end{compactenumi}
\end{lemma}
\remove{%
   \begin{proof}
       Fix an arbitrary vertex $r' \in \PS$ and compute a
       shortest-path tree $\Tree$, rooted at $r'$, of the graph
       $\GA$. Notice, that for any vertex $u \in \PS$, the unique path
       $\pi(r',u)$ from $r'$ to $u$ in $\Tree$ is a shortest path in
       $\GA$. Thus, for two points $u,v \in \PS$, the set of points
       along the two shortest paths $\pi(r,u)$ and $\pi(r,v)$, where
       $r$ is the lowest common ancestor of $u$ and $v$ in $\Tree$, is
       a good candidate for the separator set $\PC$.

       Observe, that any edge $uv$ of $\GA$ that is not part of
       $\Tree$, forms a cycle with the edges of $\Tree$. The dual
       graph of $\GA$, where we remove the edges of $\Tree$, is a tree
       whose vertices have degrees at most~$3$.  In particular, it has
       a separator edge. Taking this edge in the primal, together with
       its shortest path in $\Tree$, induces the desired cycle
       $\PC$. This cycle partition $\GA$ into three sets
       $\PA, \PB, \PC$, with $\PA$ being the points outside of the
       cycle, $\PB$ being the points inside the cycle, and $\PC$ being
       the points contained on the cycle. Since $\GA$ is planar graph,
       there is no edge between $\PA$ and $\PB$. See \cite[Lemma
       2]{lt-stpg-79} for details.
   \end{proof}
}

\begin{definition}
    For a metric space $(\FMS,\DistChar)$ and a parameter $r$, an
    $r$-net $\Net$ is a maximal set of points in $\FMS$ satisfying:
    \begin{compactenumi}
        \item For any two net points $\pp,\pq \in \Net$, $p\ne q$, we
        have $\dMY{\pp}{\pq} > r$.
        \item For any $\pp \in \FMS$,
        $\distSetY{\pp}{\Net} = \min_{ \pq \in \Net} \dMY{\pp}{\pq}
        \leq r$.
    \end{compactenumi}
\end{definition}

A net can be computed by repeatedly picking the point furthest away
from the current net $\Net$, and adding it to the net if this distance
is larger than $r$, and stopping otherwise. We denote a net computed
by this algorithm by $\netY{\FMS}{r}$.

The following lemma testifies that if we restrict the net to lay along
a shortest path in the graph, locally the cover it induces has depth
as if the graph was one dimensional.

\begin{lemma}
    \lemlab{depth:balls}%
    Let $\G$ be a weighted graph, and let $\DistChar$ be the shortest
    path metric it induces. Let $\pi$ be a shortest path in $\G$ and
    let $\Net = \netY{\pi}{r} \subseteq \pi$ be a net computed for
    some distance $r>0$.  For some $R > 0$, consider the set of balls
    $\BallSet = \Set{\ballY{\pp}{R}}{\pp \in \Net}$. For any point
    $\pq \in \VX{\G}$, we have that the degree of $\pq$ in $\BallSet$
    is at most $2R /r +1$.
\end{lemma}
\begin{proof}
    Let $\pp_1,\ldots, \pp_k$ be the points of
    $\Net \cap \ballY{\pq}{R}$ sorted by their order along $\pi$ --
    these are the only points that their balls in $\BallSet$ would
    contain $\pq$.  By the definition of the net,
    \begin{math}
        \dMY{\pp_i}{\pp_{i+1}} > r
    \end{math}
    for all $i$. Since $\pi$ is a shortest path, we also have that
    \begin{equation*}
        2R%
        \geq%
        \diamX{\ballY{\pq}{R}} %
        \geq%
        \dMY{\pp_1}{\pp_{k}} %
        =%
        \sum_{i=1}^{k-1} 
        \dMY{\pp_i}{\pp_{i+1}} %
        > (k-1)r.
        \SubmitVer{\qedhere}%
    \end{equation*}
\end{proof}

\paragraph*{Construction}
\seclab{planar_cover_constr} Let $\eps \in (0,1)$ be an input
parameter, and let $\G$ be a weighted planar graph. We assume that
$\G$ is triangulated, as otherwise it can be triangulated (we also
assume that we have its planar embedding). Any new edge $\pp\pq$ is
assigned as weight the distance between its endpoints in the original
graph. This can be done in linear time. As usual, we assume that the
minimum distance in $\G$ is one, and its spread is $\SpreadC$.

Let $\PC$ be the cycle separator given by \lemref{separator} made out
of two shortest paths $\pi_1$ and $\pi_2$.  Let $\pp_1, \pp_2,\pp_3$
be the endpoints of these two paths.

For $i=0,\ldots, m = \ceil{\smash{\log_{1+\eps/8} \SpreadC}}$, let
$\Net_i = \netY{\pi_1}{\eps r_i/8} \cup \netY{\pi_2}{\eps r_i/8} \cup
\{ \pp_1, \pp_2, \pp_3\}$, where $r_i = (1+\eps/8)^i$.  The associated
set of balls is
\begin{equation*}
    \BallSet_i = \Set{\ballY{\pp}{(1+\eps/8)r_i}}{ \pp \in \Net_i}.
\end{equation*}
The resulting set of balls is $\BallSet(\PC) = \bigcup_i
\BallSet_i$. We add the sets of $\BallSet(\PC)$ to the cover, and
continue recursively on the connected components of $\G - \PC$. Let
$\Cover$ denote the resulting cover.

\paragraph*{Analysis}

\begin{lemma}
    \lemlab{planar:cover}%
    For any two vertices $\pp, \pq \in \VX{\G}$, there exists a
    cluster $\cl \in \Cover$, such that $\pp,\pq \in \cl$, and
    \begin{math}
        \diamX{\cl} \leq (2+\eps)\dGY{\pp}{\pq}.
    \end{math}
    That is, $\Cover$ is a $(2+\eps)$-cover of $\G$.
\end{lemma}
\begin{proof}
    Assume, for the simplicity of exposition, that $\pp$ and $\pq$ get
    separated in the top level of the recursion (otherwise, apply the
    argument to the inductive step in which they get separated).  The
    shortest path between $\pp$ and $\pq$ must intersect the separator
    $\PC$, say at a vertex $v$. Assume that
    $\dGY{\pp}{v} \geq \dGY{v}{\pq}$ and that
    $r_{j-1} \leq \dGY{\pp}{v} \leq r_{j}=(1+\eps/8)^{j}$. There is a
    point $u \in \Net_{j}$ within distance $\eps r_{j}/8$ from $v$. As
    such,
    \begin{equation*}
        \dGY{\pp}{u}%
        \leq %
        \dGY{\pp}{v} + \dGY{v}{u}
        \leq %
        (1+\eps/8)^j + (\eps/8)(1+\eps/8)^j %
        \leq%
        (1+\eps/8)^{j+1},
    \end{equation*}
    which implies that $\pp \in \ballC = \ballY{u}{(1+\eps/8) r_j}$. A
    similar argument shows that $\pq$ is also in
    $\ballC$. Furthermore, we have that
    \begin{math}
        \dMY{\pp}{\pq} \geq \dMY{\pp}{v} \geq r_{j-1}.
    \end{math}
    Note that $\ballC \in \BallSet_j \subseteq \Cover$. We have that
    \begin{equation*}
        \frac{\diamX{\ballC}}{\dMY{\pp}{\pq}}%
        \leq%
        \frac{2 (1+\eps/8)r_j}{r_{j-1}}%
        =%
        2(1+\eps/8)^2%
        \leq %
        2+\eps.
        \SubmitVer{\qedhere}%
    \end{equation*}
\end{proof}

\begin{lemma}
    \lemlab{d:c:c}%
    The depth of $\Cover$ is $\Of( \eps^{-2} \log n \log \SpreadC)$.
\end{lemma}
\begin{proof}
    Fix a vertex $\pp$.  By \lemref{depth:balls}, for each $i$, at
    most
    \begin{equation*}
        3 + 2 \pth{ \frac{ 2 (1+\eps/8)r_i}{ \eps r_i /8} + 1 }%
        =%
        \Of(1/\eps).
    \end{equation*}
    balls of $\BallSet_i$ contains $\pp$. The number of such sets is
    $\Of( \log_{1+\eps/8} \SpreadC) = \Of(\eps^{-1} \log
    \SpreadC)$. The vertex $\pp$ get sent down to at most one
    recursive subproblem, and the recursion depth is $\Of(\log n)$.
    It follows that the depth of any point (and the degree of $\pp$
    specifically) is at most $\Of(\eps^{-2} \log \SpreadC \log n)$.
\end{proof}

\begin{theorem}
    \thmlab{cover:planar}%
    Let $\G$ be a weighted planar graph over $n$ vertices with spread
    $\SpreadC$. Then, given a parameter $\eps \in (0,1)$, one can
    construct a $(2+\eps)$-cover of $\G$ with depth
    $\Of( \eps^{-2} \log n \log \SpreadC)$ in polynomial time.
\end{theorem}

\begin{remark}
    It is possible to generalize \thmref{cover:planar} to the shortest
    path metric on families of graphs excluding a fixed
    minor. Specifically, by~\cite[Lemma~3.3]{klmn-mdnem-05}, there
    exists $O(s^2)$-cover of depth $O(3^s \log\SpreadC)$ for every
    metric space supported on a graph excluding $K_s$ minor and spread
    $\SpreadC$.  It may be possible to improve the approximation
    parameter to $O(s)$ using~\cite{aggnt-crtsp-19}.  This approach
    does not have a $\log n$ factor in the depth parameter, but it can
    not provide a $(2+\varepsilon)$-approximation as in
    \thmref{cover:planar}.  As communicated to us by an anonymous
    referee, the approach used here to prove \thmref{cover:planar} can
    also be extended to any family of graphs excluding fixed minor and
    obtain $(2+\eps)$-cover using the shortest paths separators
    of~\cite{ag-olups-06}.  We have not pursued those avenues.
\end{remark}


\section{From covers to reliable spanners}
\seclab{covers:to:spanners}

\subsection{The oblivious construction}

\begin{lemma}
    \lemlab{cover:to:r:s:rand}%
    Let $\Metric = (\PS, \DistChar)$ be a finite metric space, and
    suppose there exists a $\dst$-cover $\Cover$ of $\Metric$ of size
    $s$ and depth $\Depth$.  Then, there exists an oblivious
    $\epsR$-reliable $2$-hop $2\dst$-spanner for $\Metric$, of size
    \begin{math}
        \Of\bigl( s\frac{ \Depth}{\epsR} \log\frac{\Depth}{\epsR}
        \bigr).%
    \end{math}
\end{lemma}
\begin{proof}
    For each cluster $\CSet \in \Cover$, let $\StarR_\CSet$ be a
    random constellation graph on $\CSet$ as defined in
    \secref{unif-rel-star}, with reliability parameter
    $\epsRA = \epsR/ \Depth$. The resulting spanner is the union
    $\bigcup_{\CSet \in \Cover} \StarR_\CSet$.  The number of edges in
    the resulting graph is at most
    \begin{equation*}
        \sum_{C \in \Cover} 
        \Of( \cardin{C} \epsRA^{-1} \log\epsRA^{-1})
        =%
        \Of( s  \epsRA^{-1} \log\epsRA^{-1})        
        =%
        \Of\Bigl(  s \frac{\Depth}{\epsR} \log\frac{\Depth}{\epsR} \Bigr).
    \end{equation*}

    Fix an attack set $\BSet\subset \PS$.  A cluster $\CSet\in\Cover$
    is \emphw{failed} if
    $\cardin{C} \le (1+\epsRA) \cardin{C \cap \BSet } $.  Denote the
    set failed clusters by $\FlSet$.
    For $C\in \Cover \setminus \FlSet$, let $\dmgY{C}{\BSet}$ be the
    (random) set of damaged points in $C$ as defined in the proof of
    \lemref{reliable-star}, i.e., $\dmgY{C}{\BSet}=C$ if $\BSet$
    contains all the constellation's centers, and
    $\dmgY{C}{\BSet}=\BSet$ otherwise.  by \lemref{reliable-star},
    \begin{math}
        \Ex{|\dmgY{C}{\BSet}|} \le (1+\epsRA) \cardin{\BSet \cap C}.
    \end{math}
    The damaged set is defined as
    \begin{equation*}
        \EBSet%
        =%
        \Bigl(\bigcup_{C\in\FlSet} C \Bigr ) \cup
        \Bigl(\bigcup_{C\in\Cover \setminus \FlSet}
        \dmgY{C}{\BSet} \Bigr ).        
    \end{equation*}
    We next bound expected size of the \emphw{loss}
    $\EBSet \setminus \BSet$:
    \begin{align*}
      \Ex{\cardin{\EBSet\setminus \BSet} }%
      &\leq%
        \Ex{\Bigl. \smash{\sum_{C\in\FlSet}} |C\setminus \BSet|} +
        \sum_{C\in\Cover \setminus \FlSet}
        \Ex{\cardin{\dmgY{C}{\BSet}\setminus \BSet} \Bigr.\!\,} 
      \\&%
      \leq%
      \epsRA \sum_{C\in\FlSet} |\BSet \cap C| + \epsRA
      \sum_{C\in\Cover \setminus \FlSet} |\BSet \cap C| = \epsRA
      \sum_{\pp\in \BSet} \sum_{C\in \Cover}1_{\pp\in C} \le \epsRA
      \Depth |B| = \epsR |B|.
    \end{align*}

    Finally, for any two points $\pp, \pq \in \PS \setminus \EBSet$,
    there exists a non-failed cluster $C \in \Cover$ that contains
    both points, such that $\diamX{C} \leq \dst \dMY{\pp}{\pq}$. As
    such, the two hops in the resulting graph are going to be of
    length at most $ 2 \diamX{C} \leq 2 \dst \dMY{\pp}{\pq}$.
\end{proof}

\subsection{The deterministic construction}

\begin{lemma}
    \lemlab{covers:to:det}%
    Let $\Metric = (\PS, \DistChar)$ be a finite metric space over $n$
    points, and let $\Cover$ be a $\dst$-cover of it of depth $\Depth$
    and size $s$. Then, for any integer $t\geq 1$, there exists:
    \begin{compactenumA}
        \NotSubmitVer{\smallskip}%
        \item \itemlab{c:t:d:a} A $\epsR$-reliable $(2t-1)$-hop
        $(2t-1)\dst$-spanner for $\Metric$, of size
        $\Of(\epsR^{-2} \Depth^2 s n^{1/t})$.

        \NotSubmitVer{\smallskip}%
        \item \itemlab{c:t:d:b} A $\epsR$-reliable $2t$-hop
        $2t\dst$-spanner for $\Metric$, of size
        $\Of(\epsR^{-1} \Depth s n^{1/t})$.
    \end{compactenumA}
\end{lemma}
\begin{proof}
    (A) For a cluster $\cl \in \Cover$, let $\G(\cl)$ be a
    $\epsRA$-reliable spanner constructed using
    \corref{det:uniform:r:s:2t-1} on $\cl$, with the reliability
    parameter $\epsRA = \epsR/\Depth$.  Let $\G=(\PS,E)$ be a graph
    whose edge set $E$ is the union of the edge sets of $\G(\cl)$ for
    $\cl \in \Cover$.

    Let $n_i$ be the size of the $i$\th cluster, for
    $i\in \{1,\ldots, m = \cardin{\Cover}\}$.  The number of edges in
    $\G$ is bounded by
    $\Of( \epsRA^{-2} N) = \Of(\Depth^2 \epsR^{-2} N )$, where
    \begin{equation*}
        N%
        =%
        \sum_{i=1}^m n_i^{1+1/t}
        =%
        \sum_{i=1}^m n_i^{1/t} n_i 
        \leq %
        (\max_i n_i^{1/t}) \sum_{i=1}^m  n_i
        \leq%
        n^{1/t} s,%
    \end{equation*}
    since $\sum_i n_i = s$, and $\max_i n_i \leq n$. We conclude that
    the total number of edges in $\G$ is
    $\Of(\epsR^{-2} \Depth^2 s n^{1/t})$.

    Let $\BSet\subset \PS$ be an attack set.  The damage set
    $\EBSet\supseteq \BSet$ is constructed as in the proof of
    \lemref{cover:to:r:s:rand}.  Thus, as argued there,
    $\cardin{\EBSet} \le (1+\epsR)\cardin{\BSet}$.  For every two
    points $\pp,\pq \in \PS \setminus \EBSet$, there exists a cluster
    $\cl \in \Cover$, such that $\pp, \pq \in \cl$, and
    \begin{equation*}
        \diamX{ \cl }/\dst \leq  \dMY{\pp}{\pq} \leq \diamX{ \cl}.
    \end{equation*}
    Furthermore, by \corref{det:uniform:r:s:2t-1}, we have
    \begin{equation*}
        \dMZ{\G \setminus \BSet}{\pp}{\pq}
        \leq \dMZ{\G(\cl) \setminus\BSet}{\pp}{\pq}%
        \leq%
        (2t-1)\diamX{\cl}%
        \leq%
        (2t-1)\dst \cdot \dMY{\pp}{\pq}.
    \end{equation*}
    and shortest path in $\G(\cl)\setminus \BSet$ realizing
    $\dMZ{\G(\cl) \setminus\BSet}{\pp}{\pq}$ has at most $2t-1$ hops.
    \NotSubmitVer{\smallskip}
    
    (B) Similar to the above, except for using
    \corref{det:uniform:r:s:2t} instead of
    \corref{det:uniform:r:s:2t-1}.
\end{proof}


\subsection{Applications}

\subsubsection{General metrics}

\begin{lemma}
    \lemlab{r:s:metric:rand}%
    Let $\Metric = (\PS, \DistChar)$ be an $n$-point metric space of
    spread at most $\SpreadC$, and let $\epsR \in (0,1)$ and $k\in\NN$
    be parameters. Then, one can build an oblivious $\epsR$-reliable
    $\Of( k )$-spanner for $\Metric$ with
    \begin{equation*}
        \Of\Bigl(\epsR^{-1} kn^{1+1/k} \log^2 \SpreadC
        \log\frac{kn^{1/k}\log \SpreadC}{\epsR} \Bigr)
    \end{equation*}
    edges.  In particular, for $k=\log n$, we obtain a
    $\epsR$-reliable $\Of( \log n)$-spanner for $\Metric$ with
    \begin{equation*}
        \Of\bigl( \epsR^{-1} n \log n \log^2 \SpreadC (\log\log n+
        \log\log \SpreadC+\log \epsR^{-1} \bigr)        
    \end{equation*}
    edges.
\end{lemma}
\begin{proof}
    By \lemref{cover:metric:spread:b}, $\Metric$ has a
    $\dst = \Of(k)$-cover of size $\Of(n^{1+1/2k} \log \SpreadC)$ and
    depth $\Depth = \Of(k n^{1/2k} \log \SpreadC)$.  Plugging this
    into \lemref{cover:to:r:s:rand} yields an oblivious
    $\epsR$-reliable $\Of(k)$-spanner with
    $\Of\bigl(\epsR^{-1} kn^{1+1/k} \log^2 \SpreadC
    \log\frac{kn^{1/k}\log \SpreadC}{\epsR} \bigr)$ edges.
\end{proof}

\begin{lemma}
    \lemlab{r:s:metric:det}%
    Let $\Metric = (\PS, \DistChar)$ be a finite metric over $n$
    points of spread $\SpreadC$, and let $\epsR \in (0,1)$ and
    $k,t \in \NN$ be parameters. Then, one can build a
    $\epsR$-reliable $\Of( kt)$-spanner for $\Metric$ with
    $\Of(\epsR^{-1} k n^{1+1/k+1/t} \log^2 \SpreadC )$ edges.  In
    particular, when $t=\log n$, we obtain a $\epsR$-reliable
    $\Of( k \log n)$-spanner for $\Metric$ with
    \begin{equation*}
        \Of(\epsR^{-1} k n^{1+1/(2k)} \log^2 \SpreadC )
    \end{equation*}
    edges, and when $t=k$ we obtain $\epsR$-reliable
    $\Of(t^2)$-spanner for $\Metric$ with
    $\Of(\epsR^{-1} t n^{1+1/t} \log^2 \SpreadC )$ edges.
\end{lemma}
\begin{proof}
    By \lemref{cover:metric:spread:b}, $\Metric$ has a $\Of(k)$-cover
    of size $s= \Of(n^{1+1/2k} \log \SpreadC)$ and depth
    $\Depth = \Of(k n^{1/2k} \log \SpreadC)$.  Plugging this into
    \lemref{covers:to:det} \itemref{c:t:d:b}, yields a
    $\epsR$-reliable $\Of(kt)$-spanner with the number of edges
    bounded by
    \begin{equation*}
        \Of(\epsR^{-1} \Depth s n^{1/t})%
        =%
        \Of(\epsR^{-1} k n^{1+1/k+1/t}\log^2 \SpreadC).
        \SubmitVer{\qedhere}
    \end{equation*}
\end{proof}


\subsubsection{Ultrametrics}

\begin{lemma}
    \lemlab{r:s:ultra:rand}%
    Let $\Metric = (\PS, \DistChar)$ be an ultrametric over $n$ points
    with spread $\SpreadC$, and let $\epsR,\eps \in (0,1)$ be
    parameters. Then, one can build an \emph{oblivious}
    $\epsR$-reliable $(2+\eps)$-spanner for $\Metric$ with
    \begin{equation*}
        \Of\bigl( \epsR^{-1} \eps^{-2} n \log^2 \SpreadC
        \log\frac{\log\SpreadC}{\epsR\eps } \bigr)        
    \end{equation*}
    edges.
\end{lemma}

\begin{proof}
    By \corref{h:s:t:cover}, one can build a $(1+\eps/2)$-cover of
    $\Metric$ of depth $\Depth = \Of( \eps^{-1} \log \SpreadC )$ and
    size $\Of(n \Depth)$.  Plugging this into
    \lemref{cover:to:r:s:rand} yields an oblivious $\epsR$-reliable
    $2$-hop $(2+\eps)$-spanner for $\Metric$, of size
    \begin{equation*}
        \Of\Bigl( \frac{n \Depth^2}{\epsR} \log\frac{\Depth}{\epsR}
        \Bigr)%
        =%
        \Of\Bigl( \epsR^{-1} \eps^{-2} n \log^2 \SpreadC
        \log\frac{\log\SpreadC}{\epsR\eps } \Bigr).
        \SubmitVer{\qedhere}
    \end{equation*}
\end{proof}

\begin{lemma}
    \lemlab{r:s:ultra:det}%
    Let $\Metric = (\PS, \DistChar)$ be an ultrametric over $n$ points
    with spread $\SpreadC$, and let $\epsR,\eps \in (0,1)$, and
    $t \in \mathbf{N}$ be parameters. Then, one can build a
    $\epsR$-reliable $((2+\eps)t-1)$-spanner for $\Metric$ of size
    \begin{math}
        \Of(\epsR^{-2} \eps^{-3} t \cdot n^{1+1/t} \log^3 \SpreadC ).
    \end{math}
\end{lemma}

\begin{proof}
    By \corref{h:s:t:cover}, one can build a $\dst = (1+\eps/2)$-cover
    of $\Metric$ of depth $\Depth = \Of( \eps^{-1} \log \SpreadC)$ and
    size $\Of(n\Depth)$.  Plugging this into \lemref{covers:to:det}
    \itemref{c:t:d:a} yields a deterministic $\epsR$-reliable
    $(2t-1)\dst$-spanner for $\Metric$, of size
    \begin{equation*}
        \Of(\epsR^{-2} t \Depth^3 n^{1+1/t})
        =%
        \Of(\epsR^{-2} \eps^{-3} t \cdot n^{1+1/t} \log^3 \SpreadC ).
        \SubmitVer{\qedhere}
    \end{equation*}
\end{proof}


\subsubsection{Tree metrics}

\begin{lemma}
    \lemlab{r:s:tree:rand}%
    Let $\Metric = (\PS, \DistChar)$ be a tree metric over $n$ points
    with spread $\SpreadC$, and let $\epsR,\eps \in (0,1)$ be
    parameters. Then, one can build an \emph{oblivious}
    $\epsR$-reliable $(3+\eps)$-spanner for $\Metric$ with
    \begin{equation*}
        \Of\bigl( \epsR^{-1} \eps^{-2} n \polylog(n, \SpreadC) \bigr)        
    \end{equation*}
    edges, where
    $\polylog( n, \SpreadC) = \log^2 n \log^2 \SpreadC
    \log\frac{\log\SpreadC \log n}{\epsR\eps }$.
\end{lemma}

\begin{proof}
    By \lemref{cover:trees}, one can build a $(2+\eps/2)$-cover of
    $\Tree$ of depth $\Depth = \Of( \eps^{-1} \log \SpreadC \log n)$
    and size $\Of(n \Depth)$.  Plugging this into
    \lemref{cover:to:r:s:rand} yields an oblivious $\epsR$-reliable
    $2$-hop $(4+\eps)$-spanner for $\Metric$, of size
    \begin{equation*}
        \Of\Bigl( \frac{n \Depth^2}{\epsR} \log\frac{\Depth}{\epsR}
        \Bigr)%
        =%
        \Of\Bigl( \epsR^{-1} \eps^{-2} n \log^2 n \log^2 \SpreadC
        \log\frac{\log\SpreadC \log n}{\epsR\eps } \Bigr).
    \end{equation*}

    To get the improved bound on the dilation, let $\pp, \pq \in \PS$
    be two points and let $\cl \in \Cover$ be the cluster such that
    $\pp,\pq \in \cl$ and $\diamX{\cl} \leq (2+\eps) \dMY{\pp}{\pq}$.
    Assume that $|C|>2$ (otherwise $C=\{\pp,\pq\}$ and there is
    nothing to prove).  By construction, for the separator node
    $\sep \in \cl$, we have $\dMY{\sep}{\pz} \leq \diamX{\cl}/2$ for
    all $\pz \in \cl$. Observe, that the (shortest) path between $\pp$
    and $\pq$ in $\Tree$ passes through $\sep$. Thus, using the
    triangle inequality, the length of a $2$-hop path between $\pp$
    and $\pq$ via $\pz \in \cl$ can be bounded by
    \begin{equation*}
        \dMY{\pp}{\pz} + \dMY{\pz}{\pq}
        \leq
        \dMY{\pp}{\sep} + 2\dMY{\sep}{\pz} + \dMY{\sep}{\pq}
        \leq
        \dMY{\pp}{\pq} + 2 \diamX{\cl} / 2
        \leq
        (3+\eps) \dMY{\pp}{\pq}.
        \SubmitVer{\qedhere}
    \end{equation*}
\end{proof}

\begin{lemma}
    \lemlab{r:s:tree:det}%
    Let $\Metric = (\PS, \DistChar)$ be a tree metric over $n$ points
    with spread $\SpreadC$, and let $\epsR,\eps \in (0,1)$, and
    $t \in \mathbf{N}$ be parameters. Then, one can build a
    $\epsR$-reliable $((4+\eps)t-3)$-spanner for $\Metric$ of size
    \begin{math}
        \Of(\epsR^{-2} \eps^{-3} \cdot n^{1+1/t} \log^3 n \log^3
        \SpreadC ).
    \end{math}
\end{lemma}

\begin{proof}
    By \lemref{cover:trees}, one can build a $\dst = (2+\eps/2)$-cover
    of $\Tree$ of depth
    $\Depth = \Of( \eps^{-1} \log \SpreadC \log n)$ and size
    $\Of(n\Depth)$.  Plugging this into \lemref{covers:to:det}
    \itemref{c:t:d:a} yields a deterministic $\epsR$-reliable
    $(2t-1)\dst$-spanner for $\Metric$, of size
    \begin{equation*}
        \Of(\epsR^{-2} \Depth^3 n^{1+1/t})
        =%
        \Of(\epsR^{-2} \eps^{-3}  \cdot n^{1+1/t} \log^3 n \log^3 \SpreadC ).
    \end{equation*}


    To get the improved bound on the dilation, consider a $(2t-1)$-hop
    path $\pp=\pp_0,\pp_1,\dots,\pp_{2t-1}=\pq$, such that
    $\pp_i \in \cl$, for all $i$, for some cluster $\cl$. We have
    \begin{equation*}
        \begin{split}
          \sum_{i=0}^{2t-2} \dMY{\pp_i}{\pp_{i+1}} & =
                                                     \dMY{\pp}{\pp_{1}} + \dMY{\pp_{2t-2}}{\pq} +
                                                     \sum_{i=1}^{2t-3}
                                                     \dMY{\pp_i}{\pp_{i+1}} \\
                                                   & \leq \dMY{\pp}{\sep} + \dMY{\sep}{\pp_{1}} +
                                                     \dMY{\pp_{2t-2}}{\sep}
                                                     + \dMY{\sep}{\pq} + (2t-3) \diamX{\cl} \\
                                                   & \leq%
                                                     \dMY{\pp}{\pq} + (2t-2) \diamX{\cl} \leq (1 + \dst(2t-2))
                                                     \cdot \dMY{\pp}{\pq}
        \end{split}
    \end{equation*}
    for the length of the path. Thus, by using $2+\eps/2$ for the
    cover quality, we get
    \begin{equation*}
        1+(2t-2)(2+\eps/2) = (4+\eps)t - 3 - \eps \leq (4+\eps)t -3.
        \SubmitVer{\qedhere}
    \end{equation*}
\end{proof}


\subsubsection{Planar graphs}

\begin{lemma}
    \lemlab{r:s:planar:rand}%
    Let $\G$ be a weighted planar graph with $n$ vertices and spread
    $\SpreadC$. Furthermore, let $\epsR,\eps \in (0,1)$ be
    parameters. Then, one can build an oblivious $\epsR$-reliable
    $(3+\eps)$-spanner for $\G$ with
    $\Of\bigl( \epsR^{-1} \eps^{-4} n \polylog(n, \SpreadC) \bigr)$
    edges, where
    \begin{math}
        \polylog(n, \SpreadC) = \log^2 n \log^2 \SpreadC
        \log\frac{\log\SpreadC \log n}{\epsR\eps }
    \end{math}.
\end{lemma}
\begin{proof}
    By \thmref{cover:planar}, one can build a $(2+\eps/2)$-cover of
    $\G$ of depth $\Depth = \Of( \eps^{-2} \log \SpreadC \log n)$ and
    size $\Of(n \Depth)$.  Plugging this into
    \lemref{cover:to:r:s:rand} yields an oblivious $\epsR$-reliable
    $2$-hop $(4+\eps)$-spanner for $\G$, of size
    \begin{math}
        \Of\Bigl( \frac{n \Depth^2}{\epsR} \log\frac{\Depth}{\epsR}
        \Bigr)%
        =%
        \Of\Bigl( \epsR^{-1} \eps^{-4} n \log^2 n \log^2 \SpreadC
        \log\frac{\log\SpreadC \log n}{\epsR\eps } \Bigr).
    \end{math}

    We next show the improved bound on the dilation.  Using the
    notation from \lemref{planar:cover}, for a pair of points
    $\pp, \pq \in \PS$, there is a cluster $\cl \in \Cover$, such that
    $\pp, \pq \in \cl$ and $\diamX{\cl} \leq (2+\eps)\dMY{\pp}{\pq}$.
    Let $v$ be the point where the cycle separator and the shortest
    path between $\pp$ and $\pq$ intersect. Notice, that for the
    center point $u \in \cl$ we have
    $\dMY{u}{\pz} \leq \diamX{\cl} /2$, for all $\pz \in \cl$.
    Furthermore, using the notation from the proof of
    \lemref{planar:cover}, for some $j$, we have
    \begin{equation*}
        \frac{\dMY{u}{v}}{\dMY{\pp}{\pq}}
        \leq
        \frac{\eps r_j /8}{r_{j-1}}
        =
        \frac{\eps}{8} \left( 1 + \frac{\eps}{8} \right)
        \leq
        \frac{\eps}{4}.
    \end{equation*}
    Thus, the length of a $2$-hop path between $\pp$ and $\pq$ via
    $\pz \in \cl$ can be bounded by
    \begin{align*}
      \dMY{\pp}{\pz} + \dMY{\pz}{\pq}
      & \leq%
        \dMY{\pp}{v} +
        2\dMY{v}{u} + 2\dMY{u}{\pz} + \dMY{v}{\pq} \leq
        \dMY{\pp}{\pq} + 2 \frac{\eps}{4}\dMY{\pp}{\pq} + 2
        \diamX{\cl} / 2%
      \\
      & \leq%
        \left( 1 + \frac{\eps}{2} + 2+\eps \right) \dMY{\pp}{\pq}
        \leq (3+ 2\eps) \dMY{\pp}{\pq}.
        \SubmitVer{\qedhere}
    \end{align*}
\end{proof}

\begin{lemma}
    \lemlab{r:s:planar:det}%
    Let $\G$ be a weighted planar graph with $n$ vertices and spread
    $\SpreadC$. Furthermore, let $\epsR,\eps \in (0,1)$ and
    $t \in \mathbf{N}$ be parameters. Then, one can build a
    deterministic $\epsR$-reliable $((4+\eps)t-3)$-spanner for $\G$ of
    size
    \begin{math}
        \Of(\epsR^{-2} \eps^{-6} \cdot n^{1+1/t} \log^3 n \log^3
        \SpreadC).
    \end{math}
\end{lemma}
\begin{proof}
    Let $\dst = 2+\eps/2$.  By \thmref{cover:planar}, one can build a
    $\dst$-cover of $\G$ with depth
    $\Depth = \Of( \eps^{-2} \log \SpreadC \log n)$ and size
    $\Of(n \Depth)$.  Plugging this into \lemref{covers:to:det}
    \itemref{c:t:d:a} yields a deterministic $\epsR$-reliable
    $(2t-1)\dst$-spanner for $\Metric$, of size
    \begin{math}
        \Of(\epsR^{-2} \Depth^3 n^{1+1/t}) =%
        \Of(\epsR^{-2} \eps^{-6} \cdot n^{1+1/t} \log^3 n \log^3
        \SpreadC ).
    \end{math}

    The improved dilation follows by using the same argument as in the
    proof of \lemref{r:s:planar:rand} and \lemref{r:s:tree:det}.~
\end{proof}





\section{Proper expanders as reliable spanners for uniform   metric}
\seclab{expander:spanner}

The purpose of this section is to prove that with the appropriate
parameters, proper expanders (as defined in \defref{proper:expander})
and edge-union of proper expanders constitute good reliable spanners
for uniform metrics.  That is, they satisfy the requirements of
\thmref{2t-2t-1-spanner}.
 
\subsection{Preliminaries}

In the following, \emphw{$(n,d)$-graph} denotes a $d$-regular,
$n$-vertex graph.  Recall that $\lambda(G)$ denotes the normalized
eigenvalue of $\G$ -- that is, the second largest in absolute value
(see~\defref{proper:expander}).

For a given graph $\G=(\VV,\Edges)$ and $S,H\subseteq \VV$, denote by
$\Gamma_H(S)=\Set{v\in H}{ \exists u\in S,\; uv\in E}$ the
\emphw{neighbors} of $S$ in $H$.  For $S,T\subset \VV$, denote
$\EdgesY{S}{T}=\Set{ uv \in E}{ u\in S,\; v\in T}$.

The following is well known result on expanders, attributed to Alon
and Chung \cite{ac-eclst-88} in~\cite[Sec.~2.4]{hlw-ega-06}.
\begin{lemma}[Expander mixing lemma]
    Let $\GDef$ be an $(n,d)$ graph.  Then for every $S,T\subset \VV$,
    \begin{equation}
        \eqlab{mixing}%
        \cardin{ |\EdgesY{S}{T}|- \frac{d|S|\,|T|}{n} }%
        \leq%
        \lambda(G) d \sqrt{|S|\,|T|}.
    \end{equation}

\end{lemma}

We also need the following lemma which is a minor variant of known
constructions.
\begin{lemma}[\cite{bho-spda-19}]
    \lemlab{expander}%
    Let $L,R$ be two disjoint sets, with a total of $n\in 2\NN$
    elements, and let $\epsA \in (0,1)$ be a parameter.  There exists
    a bipartite graph $\G = (L \cup R, E)$ with $\Of(n/\epsA^2)$
    edges, such that \NotSubmitVer{\smallskip}%
    \begin{compactenumI}
        \item for any subset $X \subseteq L$, with
        $\cardin{X} \geq \epsA |L|$, we have that
        $\cardin{\NbrX{X}} > (1-\epsA)|R|$, and
        \NotSubmitVer{\smallskip}%
        \item for any subset $Y \subseteq R$, with
        $\cardin{Y} \geq \epsA |R|$, we have that
        $\cardin{\NbrX{Y}} > (1-\epsA)|L|$.
    \end{compactenumI}
\end{lemma}

\begin{remark}
    \remlab{constructive}%
    The randomized construction of \lemref{expander} succeeds with
    probability $1-1/n^{O(1)}$. Since we use the construction below on
    sets that are polynomially large (i.e., $n^{1/t}$), one can use
    \lemref{expander} constructively in this case (potentially losing
    an additional $\log t$ factor).  This also applies to the other
    expander constructions used in this paper.  But while the
    randomized construction works with high probability, verifying it
    seems computationally intractable.
\end{remark}

\subsection{Construction of reliable spanners from proper expanders}
\seclab{spanner:construction}

\thmref{2t-2t-1-spanner} states the existence two different spanners
and accordingly, we present two different graphs,
$\GP _{n,\epsR,2t-1}$ and $\GP _{n,\epsR,2t}$.

We begin with $\GP_{n,\epsR,2t}$.  Recall the definition of proper
expander (\defref{proper:expander}) with parameter $\delta$.  Fix
$n,t\in \NN$, $n>t$, and $\delta\in(0,1/4)$ such that
\begin{equation}%
    \eqlab{n^1/t>theta}%
    n^{1/t}\geq (\cfX{\delta}
    \delta)^{-1}.
\end{equation}
The graph $\GP_{n,\epsR,2t}$ is defined to be an $(n,d)$-graph that is
a proper expander with $\delta=\epsR$, and
\begin{equation}
    \eqlab{d}
    d = \ceil{ \max \{2 \epsR^{-1} n^{1/t}, 36C^2
       \epsR^{-3}
       \cfX{\epsR}^{-2}\}}
\end{equation}

To define $\GP_{n,\epsR,2t-1}$ we follow an idea we used slightly
inferior construction in a preliminary version of this paper, see
\cite[Section~3.3.1]{hmo-rsms-21}: Let $n'=n^{1-1/t}$.  Partition the
space to $n^{1/t}$ subsets, $A_1,\ldots,A_{n^{1/t}}$, each of size
$n'$, and let $t'=t-1$.  For every $A_i$ construct a copy of the graph
$\GP_{n',\epsR, 2t'}$ with $A_i$ being the vertices.  The degree in
those graphs is
$d'=\Of(\epsR^{-1}{n'}^{1/t'})=\Of(\epsR^{-1}n^{1/t})$.

In addition, for every $i\ne j$ connect $A_i$ with $A_j$ with a
bipartite expander $(A_i\cup A_j,\Edges_{ij})$ according to
\lemref{expander}, with $\epsA=\epsR$. \ManorHide{TO CHECK CONSTANTS}
This increases the degree by $\Of(\epsR^{-2} n^{1/t})$.  Thus, the
total degree of $\GP_{n,\epsR,2t-1}$ is $\Of(\epsR^{-2} n^{1/t})$.%
\ManorHide{Note that the second stage strictly dominates the first. It
   smells like there should be a better balance between the two that
   improve the dependence of the size of $\GP_{n,\epsR,2t-1}$ on
   $\epsR$.}

\subsection{Analysis of \TPDF{$\GP_{n,\epsR,2t}$}{G,n,epsilon,2t}}
\seclab{analysis:gp-2t}

For the rest of \secref{analysis:gp-2t}, let $\G=\GP_{n,\epsR,2t}$,
and $\lambda=\lambda(\G) $.

\subsubsection{The shadow of a bad set}

Let $\BSet \subset \VV$ an arbitrary subset, such that
$(1+5\epsR)|\BSet| < n$ (otherwise, we can choose $\EBSet=\VV$ and
there is nothing to prove).  Choose
\begin{equation}%
    \eqlab{epsilon}%
    \eps%
    =%
    (1+\epsR )(|\BSet|/n +\lambda/\sqrt{\epsR})
    \leqX{\itemref{p:3}}%
    (1+\epsR )(|\BSet|/n
    +C /\sqrt{d \epsR}) \stackrel{\Eqref{d}}\leq (1+\epsR )|\BSet|/n
    +\cepsR \epsR,
\end{equation}
so (recalling that $\delta=\epsR$, and $\cfX{\delta}=\cfX{\epsR}$)
\begin{equation}%
    \eqlab{1-delta-epsilon}%
    1-\delta-\eps%
    \geq%
    1-\epsR -|\BSet| /n -\epsR |\BSet|/n -\cepsR \epsR
    \geq%
    1-\epsR -(1-4\epsR) - \epsR 
    -\cepsR \epsR \geq \epsR.
\end{equation}

Define the ``shadows of $\BSet$'' as follows.  Let
$\SSet_0=\emptyset$, and for $i > 0$ let
\begin{equation}
    \eqlab{S:i}%
    \SSet_{i}%
    =%
    \SSet_{i-1} \cup 
    \Set{u\in \VV\setminus (\BSet\cup \SSet_{i-1})}%
    { |\EdgesY{u}{ \BSet\cup \SSet_{i-1}}| \geq \eps d \bigr.}.
\end{equation}
These are all the ``bad'' vertices that have a lot of neighbors inside
the (growing) bad set $\BSet \cup \SSet_{i-1}$.  Lastly, set
$\SSet=\cup_i \SSet_i$ the limit of $\SSet_i$, and
$\EBSet=\BSet \cup \SSet$.

\subsubsection{Bounding the size of the shadow}

By the construction above of the damaged set $\EBSet$, for every
$u\in \VV\setminus \EBSet$,
\begin{equation}
    \eqlab{E:u:EBSet}%
    |\EdgesY{u}{\EBSet}|%
    <%
    \eps d.
\end{equation}

\begin{claim} %
    \clmlab{cl:1}%
    $|\SSet|\leq \epsR |\BSet|$.
\end{claim}
\begin{proof}
    The argument we use is similar to the one used in the proof
    of~\cite[Lemma~5.3]{blmn-mrtp-05}.  Let
    $\Delta_i = \SSet_i \setminus \SSet_{i-1} $.  We have that
    \begin{align*}
      \eps d|\SSet_i|
      &\quad\leq\quad%
        \eps d \sum_{j=1}^i |\Delta_j|%
        \stackrel{\Eqref{S:i}}%
        \leq%
        \sum_{j=1}^i 
        |\EdgesY{\Delta_j}{\BSet\cup \SSet_{j-1}}|
        =
        |\EdgesY{\SSet_i}{\BSet\cup \SSet_{i-1}}|%
      \\%
      &\leqX{\Eqref{mixing}}           
        \frac{d(|\BSet|+|\SSet_{i-1}|)|\SSet_i|}{n}+ \lambda d
        \sqrt{(|\BSet|+|\SSet_{i-1}|)|\SSet_i|}.            
    \end{align*}
    This implies that
    \begin{math}
        |\SSet_i| \pth{ \eps - \frac{|\BSet|+|\SSet_{i-1}|}{n}} \leq
        \lambda \sqrt{(|\BSet|+|\SSet_{i-1}|)|\SSet_i|}.
    \end{math}
    Squaring, and dividing by $|\SSet_i|$, we have
    \begin{equation*}
        |\SSet_i| \left(\eps-\frac{|\BSet|+|\SSet_{i-1}|}{n}\right)^2
        \leq%
        \lambda^2(|\BSet|+|\SSet_{i-1}|)%
        \leq%
        \lambda^2(|\BSet|+|\SSet_{i}|),
    \end{equation*}
    which implies, for
    $\rho_i = \left(\eps-\frac{|\BSet|+|\SSet_{i-1}|}{n}\right)^2$,
    that
    \begin{equation*}
        |\SSet_i| \leq \frac{\lambda^2 |\BSet|}
        {\rho_i -\lambda^2}.
    \end{equation*}
    We claim that $|\SSet_i|\leq \epsR |\BSet|$ for every $i$.
    Indeed, otherwise let $i$ the smallest index such that
    $|\SSet_i| > \epsR |\BSet|$. So
    $|\SSet_{i-1}|\leq \epsR |\BSet|<|\SSet_i|$.  By definition, see
    \Eqref{epsilon}, we have that
    \begin{math}
        \eps%
        =%
        (1+\epsR )(|\BSet|/n +\lambda/\sqrt{\epsR})
    \end{math}
    and as such
    \begin{equation*}
        \rho_i%
        =%
        \pth{\eps-\frac{|\BSet|+|\SSet_{i-1}|}{n}}^2
        \geq%
        \pth{
           (1+\epsR )\pth{\frac{|\BSet|}{n} +
              \frac{\lambda}{\sqrt{\epsR}} } 
           -\frac{(1+\epsR)|\BSet|}{n}}^2%
        \geq%
        \frac{1+\epsR}{\epsR}
        \lambda^2
    \end{equation*}
    We thus have that
    \begin{equation*}
        \epsR |\BSet|
        <%
        |\SSet_i|
        \leq%
        \frac{\lambda^2
           |\BSet|}{\rho_i
           -\lambda^2}%
        \leq%
        \frac{\lambda^2 |\BSet|}{\frac{1+\epsR}{\epsR}\lambda^2
           -\lambda^2}%
        =%
        \epsR |\BSet|,
    \end{equation*}
    which is a contradiction.
\end{proof}

\subsubsection{The expansion happens outside the bad set}

\begin{claim}
    \clmlab{cl:2}%
    Let $U\subseteq \VV\setminus\EBSet$.  If
    $|U|\leq \cfX{\epsR} \epsR n^{1-1/t}$ then
    $|\Gamma_{\VV\setminus \EBSet}(U) |\geq n^{1/t}|U|$.
\end{claim}
\begin{proof}
    We have
    \begin{equation*}
        |U|%
        \leq%
        \cfX{\epsR} \epsR n^{1-1/t}
        = \frac{\cepsR n}{n^{1/t} / \epsR}
        \leq%
        \frac{\cepsR n}{d/2}
        \leq%
        \frac{\cepsR n}{d},
    \end{equation*}
    since
    \begin{math}
        d \geq 2 n^{1/t}/ \epsR
    \end{math}
    by \Eqref{d}.  As such, by the expansion property \itemref{p:2},
    $|\Gamma_\VV(U)|\geq (1-\delta)d |U|$ (as $\delta = \epsR$).
    Furthermore
    \begin{equation*}
        |\Gamma_{\EBSet}(U)|\leq |\EdgesY{U}{\EBSet}|
        \stackrel{\Eqref{E:u:EBSet}}\leq \eps d |U|,            
    \end{equation*}
    so
    \begin{equation*}
        |\Gamma_{\VV\setminus \EBSet}(U)|%
        \geq%
        |\Gamma_\VV(U)| - |\Gamma_{\EBSet}(U)|%
        \geq%
        (1-\delta-\eps)d |U|
        \stackrel{\Eqref{1-delta-epsilon}}\geq \epsR d |U|
        \stackrel{\Eqref{d}}\geq n^{1/t}|U|.
        \SubmitVer{\qedhere}
    \end{equation*}
\end{proof}

Let
\begin{equation*}
    \Ball_{\G\setminus \EBSet}(u,i)%
    =%
    \Set{v\in \VV\setminus \EBSet}%
    { \mathsf{d}_{\G\setminus \EBSet}(u,v)\leq i},    
\end{equation*}
the ball of radius $i$ around $u$ in the shortest path metric of the
graph $\G \setminus \EBSet$.  Define $U_0=\{u\}$, and
$U_i=\Gamma_{\VV\setminus \EBSet}(U_{i-1})$.  Observe that
$\Ball_{\G\setminus \EBSet}(u,i)=\bigcup_{j=0}^i U_j$.


\begin{claim} \clmlab{n^(1-1/t)}
    $|\Ball_{\G\setminus \EBSet} (u,t-1)| \geq n^{1-1/t}$.
\end{claim}
\begin{proof}
    Assume the contrary. Then
    $|U_{t-1}| \leq |\Ball_{\G\setminus \EBSet} (u,t-1)| < n^{1-1/t}$,
    which implies that there exists $j<t-1$ such that
    $|U_{j+1}|<|U_j| n^{1/t}$.  From \clmref{cl:2}, this means that
    $|U_j|> \cepsR \epsR n^{1-1/t}$.  So take $K\subset U_j$ such that
    $|K|=\cepsR \epsR n^{1-1/t}$. Then again by \clmref{cl:2},
    \begin{equation*}
        |\Ball_{\G\setminus \EBSet} (u,t-1)|\geq
        |\Gamma_{V\setminus \EBSet}(K)|\geq \cepsR \epsR
        n^{1-1/t} n^{1/t} \stackrel{\Eqref{n^1/t>theta}}\geq
        n^{1-1/t}.            
        \SubmitVer{\qedhere}
    \end{equation*}
\end{proof}

We summarize the relevant properties of $\GP_{n,\epsR,2t}$.

\begin{lemma}%
    \lemlab{reliably:large:subsets}%
    The graph $\GP_{n,\epsR,2t}$ has the following properties.  For
    any $\BSet \subset \VV$, there exists a set $\EBSet$,
    $\BSet \subseteq \EBSet \subset \VV$,
    $|\EBSet|\leq (1+5\epsR) |\BSet|$, such that for any
    $u\in \VV\setminus \EBSet$, we have
    \begin{align}
      \eqlab{n^(1-1/t)}
      |\Ball_{\G\setminus \EBSet}(u,t-1)| &\geq  n^{1-1/t},\\
      \eqlab{Gamma_V}
      \Gamma_V(\Ball_{\G\setminus
      \EBSet} (u,t-1)) &\geq (1-\epsR)n,\\
      \eqlab{reliably:large:subsets}
      |\Ball_{\G\setminus \EBSet}(u,t)| & \geq \epsR n.
    \end{align}
\end{lemma}
\begin{proof}
    Recall that when $(1+5\epsR)|B|\ge n$, we can take $\EBSet=V$ and
    there is nothing to prove. So assume from now on that
    $(1+5\epsR)|B|<n$.

    When $t=1$, $\GP_{n,\epsR,2t}$ is the complete graph, and we take
    $\EBSet=\BSet$.  \Eqref{n^(1-1/t)} is equiavlent to $1\ge 1$;
    \Eqref{Gamma_V} follows since $\Gamma_V(\{u\})=V$ in complete
    graph; and \Eqref{reliably:large:subsets} follows simply because
    $|B|\le n/(1+5\epsR)\le (1-\epsR)n$.

    Assume next that $t\geq 2$.  \Eqref{n^(1-1/t)} is just a
    repetition of \clmref{n^(1-1/t)}.  To prove
    \Eqref{Gamma_V}. Observe that
    \begin{equation*}
        |\Ball_{\G\setminus \EBSet} (u,t-1)|
        \stackrel{\Eqref{n^(1-1/t)}}\geq n^{1-1/t}
        \stackrel{\Eqref{d}}\geq 2n/(\epsR d) .            
    \end{equation*}
    We conclude using Property~\itemref{p:1} of proper expanders that
    \begin{equation*}
        \Gamma_V(\Ball_{\G\setminus
           \EBSet} (u,t-1))\geq (1-\epsR)n.
    \end{equation*}
    By \clmref{cl:1},
    \[|\EBSet|\leq (1+\epsR)|\BSet| \le
        (1+\epsR)|\frac{n}{1+5\epsR}\le(1-3\epsR)n .
    \]
    So,
    \begin{equation*}
        |\Ball_{\G\setminus \EBSet} (u,t)|\geq |\Gamma_{V\setminus
           \EBSet}(\Ball_{\G\setminus \EBSet} (u,t-1))| \ge
        |\Gamma_{V}(U_{t-1})|-|\EBSet|\geq \epsR n.        
    \end{equation*}
    This complete the proof of \lemref{reliably:large:subsets}.
\end{proof}

\begin{proposition}
    \proplab{2t-spanner}%
    The graph $\GP_{n,\epsR,2t}$ is a $\epsR$-reliable $2t$-spanner
    for $n$-point uniform space with $\Of(\epsR^{-1} n^{1+1/t})$
    edges.
\end{proposition}
\begin{proof}
    Fix $u,v\in \VV\setminus \EBSet$.  By~\Eqref{Gamma_V},
    \begin{equation*}
        \min\Bigl\{|\Gamma_V(\Ball_{\G\setminus \EBSet}(u,t-1))|,\,|
        \Gamma_V(\Ball_{\G\setminus \EBSet}(v,t-1))| \Bigr\} 
        \geq%
        (1-\epsR) n .        
    \end{equation*}
    So,
    \begin{align*}
      |\Ball_{\G\setminus \EBSet}(u,t)\cap \Ball_{\G\setminus
      \EBSet}(v,t)|%
      &\geq%
        |\Gamma_V(\Ball_{\G\setminus
        \EBSet}(u,t-1)) \cap \Gamma_V(\Ball_{\G\setminus
        \EBSet}(v,t-1))| -
        |\EBSet|%
      \\
      &\geq (1-2\epsR)n -(1-5\epsR)n>0.
    \end{align*}
    We conclude that
    $\Ball_{\G\setminus \EBSet}(u,t)\cap \Ball_{\G\setminus
       \EBSet}(v,t)\ne \emptyset$, which means that there is a path of
    length at most $2t$ in $G\setminus \EBSet$ between $u$ and $v$.
\end{proof}

\subsubsection{Proof of \TPDF{\thmref{reliable:vertex:expander}}%
{\nthmref{reliable:vertex:expander}}}
\seclab{r:v:e:proof}

\SubmitVer{~\bigskip}

\RestatementOf{\thmref{reliable:vertex:expander}}%
{\ThmRelVertexExpanderBody{}}


\begin{proof}
    The graphs are simply $\GP_{n,\epsR,2t}$ for $t=\log_h n$. The
    claims follow immediately from \clmref{cl:2} and
    \propref{2t-spanner}.
\end{proof}

\subsection{Analysis of \TPDF{$\GP_{n,\epsR,2t-1}$}{G}}

\begin{proposition}
    The graph $\GP_{n,\epsR,2t-1}$ is $\epsR$-reliable
    $(2t-1)$-spanner for $n$-point uniform space with
    $\Of(\epsR^{-2} n^{1+1/t})$ edges.
\end{proposition}
\begin{proof}
    Recall the definition of $\GP_{n,\epsR,2t-1}$ from
    \secref{spanner:construction}.  Let $t'=t-1$, and $n'=n^{1-1/t}$.
    Given an attack set $\BSet$, construct $B^+_i$ from
    $A_i\cap \BSet$ according to \lemref{reliably:large:subsets}, and
    define $\EBSet=\cup_i B^+_i$. Fix $u,v\in \VV\setminus \EBSet$.
    If $u,v\in A_i$ then by \propref{2t-spanner}, there is a
    $2t'=2t-2$ path between them.

    If $u\in A_i$, $v\in A_j$, $i\ne j$.  Then by
    \Eqref{reliably:large:subsets},
    $|\Ball_{A_i\setminus \EBSet}(u,t')|\geq \epsR n'$, and
    $|\Ball_{A_j\setminus \EBSet}(v,t')|\geq \epsR n'$.
    By~\lemref{expander}, there is an edge in $E_{ij}$ between
    $\Ball_{A_i\setminus \EBSet}(u,t')$ and
    $\Ball_{A_j\setminus \EBSet}(v,t')$.  and hence a path of length
    $2t'+1=2t-1$ in $G\setminus \EBSet$.
\end{proof}

\section{Concluding remarks and open problems}

\paragraph*{Subsequent work}
Recently Filtser and Le~\cite{al-rslso-21} improved some of the
results here. They obtained bounds that do not depend on the spread of
the metrics in some cases, for the oblivious adversary case.  They
also obtained reliable spanners (in the oblivious adversary model) for
trees with tight stretch of $2$ and for planar graphs with tight
stretch of $2+\eps$.

\paragraph*{Tradeoffs in deterministic constructions for general spaces}
Classical spanners are known to have an approximation--size trade-off
for general $n$-point metrics: To achieve $\Theta(t)$ approximation it
is sufficient and necessary to have $n^{1+1/t}$ edges in the worst
case.  In contrast, for reliable spanners we were only able to show an
upper bound on the trade-off, with no asymptotically matching lower
bound: To achieve $\Of(t^2)$ approximation it is sufficient to have
$\widetilde{\Of}(n^{1+1/t})$ edges.  Classically, the uniform metric
is $2$-approximated by a star graph with only $n-1$ edges.  In
contrast, we have shown here reliable spanners for uniform metric have
approximation--size trade similar to the classical spanner for general
metrics. The connection between the two problems is quite intriguing,
and is worthy of further research.

\paragraph*{The dependence of the size on the spread}

The size of spanners constructed in this paper depends on the spread
of the metric space. This is because of the reduction to uniform
spaces via covers, in which the dependence on the spread is
unavoidable in general.  However, in some setting this dependence is
avoidable.  For example~\cite{bho-spda-19,bho-srsal-20} achieves this
for fixed dimensional Euclidean spaces, and~\cite{al-rslso-21}
achieves it for doubling spaces, and general finite spaces in the
oblivious adversary model. Getting spread-free bounds for the
non-oblivious adversary is an interesting problem for further
research.

\paragraph*{Explicit constructions}
To the best of our knowledge, there is no known polynomial time
deterministic algorithm for constructing expanders with
Property~\itemref{p:1} or Property~\itemref{p:2}. Getting such a
construction is an interesting open problem.


\paragraph*{Acknowledgments}
The authors thank Kevin Buchin for useful discussions and references,
and the anonymous referees for their suggestions.


   \NotSubmitVer{%
      \bibliographystyle{alpha}%
   }%
   \SubmitVer{%
      \bibliographystyle{ACM-Reference-Format}%
   }%
   \bibliography{reliable_metrics}%


\appendix
\section{Random regular graphs}
\apndlab{proper:expander}

In \thmref{proper:expander} above we prove that random construction
using a union of random permutations yields a proper expander, see
\defref{proper:expander}.  The properties are well known and are held
by random regular graphs asymptotically almost surely.  However, the
literature on random regular graphs is more concerned with the setting
of constant $d$ and $n$ tends to infinity, especially regarding
Property~\itemref{p:3}).  Here we need a slightly different range,
where $d$ is sufficiently large, and $n\geq d^2$.  For completeness,
we gather here proofs and references that prove
\thmref{proper:expander}.

\subsection{Construction}
\seclab{construction}

We use the \emphw{permutation model} $\Gnd$ for constructing random
regular graph (see \cite{w-mrrg-99} for other models).  Assuming $d$
is even, sample $d/2$ independent and identically distributed random
permutations $\pi_1,\ldots,\pi_{d/2}\in S_n$, where $S_n$ is the set
of all permutations of $\IRX{n} = \{1,\ldots, n\}$.  The resulting
graph $\GDef$, has $\VV=\IRX{n}$ and
\begin{equation*}
    \Edges
    =%
    \Set{ \{i ,\pi_j(i)\} }{ i\in\IRX{n} \text{ and }
       j\in \IRX{d/2}\bigr.}.
\end{equation*}

\subsection{Analysis}

\subsubsection{Property \TPDF{\itemref{p:3}}{\nitemref{p:3}}}

All the proofs that random $d$-regular graphs have second eigenvalue
at most $C/\sqrt{d}$ (that we are aware of) are non-trivial.  It was
first proved in~\cite{fks-serrg-89}, and by now there are many proofs
of this.  Notably, Friedman \cite{f-pasec-03} showed that random
regular graphs are ``almost Ramanujan'', i.e.,
$\lambda \leq \frac{2\sqrt{d-1}+\eps}{d}$ with probability $1-o_n(1)$,
see also~\cite{p-ergnp-15} for a recent and simpler proof.
Unfortunately, most of those papers are interested in the settings
where the degree $d$ is constant and only the number of vertices
$n \to\infty$, which is not suitable for our needs.  However, the
argument of Kahn and \Szemeredi \cite{fks-serrg-89} does work when $d$
is allowed to tend to infinity (together with $n$). Specifically,
Dumitiriu \etal \cite{djpp-fltrr-13} used it to prove the following.

\begin{theorem}[\protect{\cite[Theorem~24]{djpp-fltrr-13}}]
    \thmlab{p:3}%
    Fix $C=41000$. There exists $K>0$ such that for any even $d$ and
    $n>\max\{K,d\}$, we have
    \begin{math}
        \displaystyle \Prob{ \lambda(\G)\leq C/\sqrt{d} \, }%
        \geq%
        1-2/n^{2},
    \end{math}
    where $\Prob{}$ is the probability in the permutation model
    $\Gnd$.
\end{theorem}

This implies \itemref{p:3} holds a.a.s. in the permutation model.
Properties~\itemref{p:1} and~\itemref{p:2} are much easier to prove
and are considered folklore.  They are usually proved in different,
more convenient, models of random regular graphs.  However, contrary
to the case when the degree is fixed, in the high-degree regime the
different random models are not necessarily contiguous, i.e.,
asymptotically almost surely properties are not necessarily equivalent
among the different models (for more on this, see~\cite{w-mrrg-99}).
Therefore, for completeness, we next provide proofs that
Properties~\itemref{p:1} and~\itemref{p:2} hold asymptotically almost
surely in the permutation model.

\subsubsection{Property \TPDF{\itemref{p:1}}{\nitemref{p:1}}}

\begin{lemma}
    \lemlab{binom:games}%
    For integers $s, t,n \geq 0$, such that $s \leq t \leq n$, we have
    \begin{math}
        \binom{t}{s} /\binom{n}{s} \leq \pth{\frac{t}{n}}^s.
    \end{math}

    Similarly, if $s \leq t\le n$ and $n \geq 3s$, we have
    \begin{math}
        \binom{t}{s} /\binom{n}{2s} \leq \pth{\frac{t}{n}}^s
        \bigl({\frac{2s}{n-s}}\bigr)^s.\Bigr.%
    \end{math}
\end{lemma}
\begin{proof}
    Observe that for $k \leq m' \leq m$, we have
    \begin{equation*}
        \frac{\binom{m'}{k}}{\binom{m}{k}}%
        =%
        \frac{m'}{m} \cdot \frac{ m'-1}{m-1} \cdots \frac{m'-k+1}{m-k+1} 
        \leq%
        \pth{\frac{m'}{m}}^k,
    \end{equation*}
    as $(m'-i)/(m-i) \leq m'/m$.  As such, using the (easy to verify)
    identity
    \begin{math}
        \binom{n}{2s}%
        = %
        \binom{n}{s} \binom{n-s}{s}/ \binom{2s}{s},
    \end{math}
    we have
    \begin{equation*}
        \frac{\binom{t}{s}}{\binom{n}{2s}}%
        =%
        \frac{\binom{t}{s}\binom{2s}{s}}{\binom{n}{s} \binom{n-s}{s}}%
        \leq%
        \pth{\frac{t}{n}}^s \pth{\frac{2s}{n-s}}^s.
        \SubmitVer{ \qedhere}        
    \end{equation*}
\end{proof}


\begin{lemma}
    \lemlab{Gamma:S:subset:T}
    Fix sets $S,T\subset \VV$, with $s=|S|$ and $t=|T|$, such that
    $|S| \leq |T|$. Then
    \begin{math}
        \Prob{\bigl.\Gamma(S)\subseteq T}%
        \leq%
        \pth{\frac{t}{n}}^{sd/2}.
    \end{math}
\end{lemma}
\begin{proof}
    Fix $T'\subseteq T$ of cardinality $|T'|=s$.  Then
    \begin{math}
        \Prob{\pi(S)= T'} = 1/ \binom{n}{s}.
    \end{math}
    By the union bound, and \lemref{binom:games}, we have
    \begin{equation*}
        \Prob{\bigl. \pi(S)\subseteq T}
        =
        \sum_{ {T' \in \binom{T}{s}}} 
        \Prob{ \pi( S)=T' \bigr.}
        =\frac{\binom{t}{s}}{\binom{n}{s}}  \leq \left(\frac tn\right)^s.
    \end{equation*}

    Let $\pi_1,\ldots,\pi_{d/2}$ be the permutations used in the
    construction. We have that
    \begin{equation*}
        \Prob{\bigl. \Gamma(S)\subseteq T}%
        \leq%
        \ProbC \Bigl[\,  \bigwedge_{k=1}^{d/2} {\pi_k}(S)\subseteq T \Bigr]
        =%
        \pth{\Prob{  {\pi}(S)\subseteq T} \Bigr.}^{d/2}
        \leq%
        \left( \frac{t}{n}\right)^{sd/2}. 
        \SubmitVer{ \qedhere}        
    \end{equation*}
\end{proof}

\begin{lemma}
    \lemlab{p:1}%
    Fix $\eps,\delta\in(0,1/2)$.  If
    $\sqrt{n} \geq d\geq 12 /(\eps \delta)$ is an integer.  Then,
    asymptotically almost surely over the random $(n,d)$-graph
    $G=(V,E)$, for any $S\subset \VV$ with $|S|\ge \eps n$, we have
    $|\Gamma(S)|\ge (1-\delta)n$. That is, property \itemref{p:1} of
    \defref{proper:expander} holds.
\end{lemma}
\begin{proof}
    Using the union bound on the bound established in
    \lemref{Gamma:S:subset:T} we conclude that the probability that
    there exists a subset $S$ of size $s=\eps n$, such that
    $|\Gamma(S)| \leq (1-\delta) n=t$, is at most
    \begin{align*}
      \binom ns \binom nt \left(\frac tn\right)^{sd/2} 
      &
        \leq%
        \pth{\frac{en}{s}}^s 2^n \left(\frac tn\right)^{sd/2} 
        \leq%
        \pth{\frac{e}{\eps}}^{\eps n} 2^n (1-\delta)^{\eps n d/2}
        \leq%
        \pth{\frac{e}{\eps}}^{\eps n} 2^n \exp\pth{-\delta \eps n d/2}
      \\&
      \leq%
      2^n \exp\pth{\eps n + \eps n \ln
      \frac{1}{\eps} -6 n }
      \leq%
      \exp\pth{ 2 n +  n/e -6 n }
      \leq \exp(-3n),
    \end{align*}
    since $\max_{x >0 } x \ln (1/x) = 1/e$, as easy calculation
    shows\footnote{Indeed, for $f(x) = x \ln (1/x)$, we have
       $f'(x)= -1 + \ln x^{-1}$. Setting $f'(x) = 0$, implies
       $1 =\ln x^{-1}$. Namely, the maximum of $f$ is achieved at
       $x= 1/e$, where $f(1/e) = 1/e$.}.
\end{proof}

\subsubsection{Property \TPDF{\itemref{p:2}}{\nitemref{p:2}}}

The argument above used only the forward edges associated with the
$d/2$ random permutations.  Since there are only $d/2$ such edges
associated with every vertex, that argument can not prove
vertex-expansion close to $d$ as is stated in \itemref{p:2}.  For this
we have to also use the backward edges associated with the permutation
as well.  Since the backward edges and the forward edges associated
with the same random permutation are not stochastically independent,
more care is needed to prove this property, as we shall now see.

\begin{claim}%
    \clmlab{in-edges}%
    Fix $\eps,\eta\in (0,1)$ such that
    $0<\eps<\frac{1}{3d}<\frac{3}{d}<\eta<1$. Then
    \begin{equation*}
        \Prob{\Bigl. \exists S\subset \VV, |S|\leq
           \eps n, |\EdgesY{S}{S}| >\eta d |S|}
        =%
        O(n^{-0.4}).        
    \end{equation*}
\end{claim}
\begin{proof}
    Fix $S\subset \VV$, of cardinality $s \leq \eps n$.  Observe that
    \begin{equation*}
        E(S,S)%
        =%
        \Set{ \{u, \pi_i(u)\}}{u \in S, \pi_i(u) \in S, \text{ and }
           i \in \IRX{d/2} \bigr.}.
    \end{equation*}
    Fix $R\subseteq S\times \IRX{ d/2}$.  For any $(u,i)\in R$, let
    $\Event_{u,i}$ be the event that $\pi_i(u)\in S$. The events
    induced by $R$, that is $\Set{\Event_{u,i}}{(u,i) \in R}$ are not
    independent, but they are non-positively correlated.  Indeed, for
    a set $R\subset S\times \IRX{d/2}$ consider the event
    $\Event_R = \bigwedge_{(z,j) \in R} \Event_{z,j}$.  For
    $(u,i)\notin R$ we have,
    \begin{equation*}
        \Prob{\pi_i(u)\in S \;\big |\; \Event_R}
        \leq%
        \Prob{\pi_i(u)\in S \bigr.}
        =%
        \frac{s}{n}.
    \end{equation*}
    Indeed, observe that all the sub-events in $\Event_R$ involving
    $j \neq i$, are irrelevant (that is, stochastically independent of
    the events $\Event_{x,i}$).  As such, let
    $X = \Set{ x }{(x,i) \in R}$. Imagine now picking the permutation
    $\pi_i$ by first picking the locations of the elements of $X$, and
    then choosing the location of $u$. Clearly, if the event
    $\Event_R$ happened, then there are only $s - |X|$ empty slots in
    $S$, and $n-|X|$ slots available overall. As such, we have that
    \begin{equation*}
        \ProbCond{\pi_i(u)\in S}{ \Event_R}
        \leq%
        \frac{s-|X|}{n - |X| } \leq \frac{s}{n}
        =
        \Prob{\pi_i(u)\in S \bigr.}.
    \end{equation*}

    Now, order the elements of $R$ in arbitrary order, and let $R(i)$
    be the subset with the first $i$ elements in that order.  We have
    \begin{equation*}
        \Prob{\Event_R\bigl.} %
        =%
        \prod_{i=1}^{|R|} \ProbCond{\Event_{R(i)} }{ \Event_{R(i-1)}}
        \leq%
        (s/n)^{|R|}.        
    \end{equation*}
    Therefore,
    \begin{align*}
      \Prob{\bigl. |E(S,S)|>\eta d s}        
      &\leq%
        \sum_{r=\eta d s +1}^{ds/2}
        \sum_{ R \in \binom{S\times \IRX{ d/2}}{r}}
        \Prob{\Event_R\bigr.}
        \leq%
        \sum_{r=\eta d s}^{ds/2} \binom {ds/2}r \cdot
        \pth{ \frac{s}{n}}^r 
        \leq%
        \sum_{r=\eta d s}^{ds/2} \pth{\frac{eds/2}{r} \cdot
        \frac{s}{n}}^r 
      \\&%
      \leq%
      \sum_{r=\eta d s}^{ds/2} \pth{\frac{eds^2 }{2r n} }^r 
      \leq%
      \left(\frac{es}{\eta n}\right)^{\eta ds},
    \end{align*}
    since the summations behaves like a geometric series, using
    \begin{math}
        \frac{eds^2}{2rn} \leq%
        \frac{eds^2}{2\eta d sn} =%
        \frac{es}{2\eta n} \leq%
        \frac{1}{2}.
    \end{math}
    Therefore
    \begin{align*}
      &\hspace{-1cm}%
        \Pr\bigl[\exists S\subset \VV, |S|\leq \eps n, |E(S,S)| >\eta d s \bigr]
        \leq%
        \sum_{s=1}^{\eps n} \binom{n}{s}
        \left(\frac{es}{\eta n}\right)^{\eta ds} 
        \leq%
        \sum_{s=1}^{\eps n}
        \left ( e^2 \left ( \frac sn\right )^{\eta d-1} \eta^{-\eta d}\right)^s
        =%
        O(n^{-0.4}).
    \end{align*}
    The last inequality is derived by partitioning the sum on $s$
    up-to $2\log n$ -- this part is bounded by
    \begin{math}
        O\bigl({\frac{\log^{3/2} n }{ \sqrt{n}}}\bigr).
    \end{math}
    The second part of the summation from $2\log n$ to $\eps n$ be
    bounded by $O\bigl((\eps/\eta)^{2\log n}\bigr)$.
\end{proof}

\begin{lemma}
    \lemlab{p:2}%
    For every $\delta\in (0,1/4)$, there exists $c_\delta>0$, such
    that for any sufficiently large integer $d\ge e^{32/\delta}$, and
    any $n\ge d^2$, asymptotically (in $n$) almost surely, a random
    $(n,d)$ graph $G=(V,E)\sim \Gnd$ have the following vertex
    expansion property: For any $S\subset \VV$ with
    $|S|\leq e^{-2.2/\delta} n/d$, we have
    $|\Gamma(S)|\ge (1-\delta)d |S|$.

    That is, Property \itemref{p:2} of \defref{proper:expander} holds.
\end{lemma}
\begin{proof}
    Fix $\eta=\delta/2$.  Define $\EventB$ to be the event, that for
    all subsets $S \subset \VV$ of cardinality at most $\eps n$, we
    have that
    \begin{equation*}
        |E(S,S)|\leq \eta d
        |S|.
    \end{equation*}
    By \clmref{in-edges}, $\Pr[\EventB]\ge 1-O(n^{-0.4})$.  So, fix
    $S\subset \VV$ of cardinality $s\leq \eps n$.  The following
    stochastic process samples a random bijection $\pi$ on $V$
    uniformly.
    \begin{compactenumi}[itemsep=0pt,ref=\roman*]
        \NotSubmitVer{\smallskip}%
        \item Order the vertices of $S$ arbitrarily $v_1,\ldots, v_s$.
        
        \item For $i=1, \ldots, s$, pick $\pi(v_i)$ uniformly at
        random from $V\setminus \{\pi(v_1),\ldots, \pi(v_{i-1})\}$.

        \item Let $U = S\cap \pi(S)$. Reorder the elements of $S$ as
        $u_1,\ldots,u_s$, such that $U =\{u_{\ell+1},\ldots, u_s\}$.
        
        \item For $i =1,\ldots, \ell$, pick $\pi^{-1}(u_i)$ uniformly
        at random from
        $V\setminus (S \cup \{\pi^{-1}(u_1),\ldots,
        \pi^{-1}(u_{i-1})\})$.

        \item \itemlab{so:far} Denote by $\pi'$ the partial injection
        constructed so far.
        
        \item \itemlab{nothing} Note that, for $i=\ell+1, \ldots, s$,
        the element $\pi^{-1}(u_i)$ is already fixed, and it is in
        $S$, as such we do nothing.

        \item Complete $\pi'$ to a full permutation by sampling a
        random bijection
        $(V\setminus (S\cup \pi'(S))) \mapsto (V\setminus (S\cup
        \pi'(S)))$
    \end{compactenumi}
    \NotSubmitVer{\medskip}%
    One can sample random regular graph $\G$ from $\Gnd$ by applying
    the above independently $d/2$ times, and constructing $d/2$ random
    permutations $\pi_1, \ldots \pi_{d/2}$.

    We next do for a fixed $S\subset \VV$ of cardinality
    $s\leq \eps n$ the following thought experiment.  In the
    construction of each of the $d/2$ permutation, we replace step
    \itempref{nothing} above with the following step generating a new
    partial injection $\tau$ into $V$, as follows:
    \begin{compactenumi}
        \NotSubmitVer{\smallskip}%
        \item[\quad\quad(\itemref{nothing}$'$)] \itemlab{nothing:2} %
        For $i = \ell+1, \ldots, s$, pick $\tau^{-1}(u_i)$ randomly
        and uniformly from %
        \begin{equation*}
            \VV\setminus (S \cup
            \{\pi^{-1}(u_1),\ldots, \pi^{-1}(u_{\ell}) ,
            \tau^{-1}(u_{\ell+1}),\ldots, \tau^{-1}(u_{i-1})\}).            
        \end{equation*}
    \end{compactenumi}
    \NotSubmitVer{\medskip}%
    Denote by $\tau_i=\tau_i^S$ that random partial injection.  Denote
    \begin{equation*}
        F_S=\Set{\bigr.\smash{\{u,{\tau^{-1}_i(u)}\}}}{ i\in\IRX{d/2}, u\in S\cap
           \pi_i(S)}
    \end{equation*}
    the set of auxiliary edges associated with the above process for
    $S$.  Denote $\GS=(V,E\cup F_S)$ -- it is the result of
    redirecting all the edges of $E(S,S)$ to be outside the
    $S\times S$ biclique.  The graph $\GS$ is randomly generated, but
    the content of $F_S$ is not measurable in $\Gnd$ -- meaning that
    the underline probability spaces is more refined than $\Gnd$.  We
    denote by $\wGnd$ a probability space from which one can sample
    all of $\GS$. I.e., it contains $d/2$ random permutations
    $\pi_1,\ldots,\pi_{d/2}$ that constitute $\Gnd$ and all the
    relevant partial injections $\tau^S_i$.

    Note that $|F_S|=|E(S,S)|$, and therefore depends only on $\Gnd$.
    Furthermore, If $\G \in \EventB$, then for every $S\subset \VV$,
    with cardinality at most $\eps n$, we have that
    $|F_S| \leq \eta d |S|$.  Hence, conditioned on $\EventB$, for any
    such $S$,
    \begin{equation}
        \eqlab{|G|<|GS|}%
        |\Gamma_G(S)|
        \geq%
        |\Gamma_{\GS}(S)| -\eta d |S|.
    \end{equation}

    We next bound for a fixed $S\subset \VV$ of cardinality
    $s\leq \eps n$, the probability that there exists a subset $T$ of
    cardinality $t$ for which $\Gamma_G(S)\subseteq T$.  Begin by
    fixing $S\subseteq T\subseteq \VV$ of cardinality $|T|=t$.  We
    claim that
    \begin{equation}
        \eqlab{pi,pi^-1}
        \Prob{\Bigl. \pi(S)\cup\pi{}^{-1}(S)\cup
           \tau^{-1}(S\cap\pi(S))\subseteq T}%
        \leq %
        \pth{\frac{t}{n}}^{2s}.
    \end{equation}

    Indeed, fix $T'\subset T$ of cardinality $s$.  Then
    \begin{equation*}
        \Prob{\Bigl.\pi(S)=T'}\leq \binom{n}{s}^{-1}.        
    \end{equation*}
    Fix $T''\subset \VV \setminus S$ of cardinality $s$.  Recall that
    we denoted by $\pi'\subset \pi$ the partial injection sampled at
    the end of step \itempref{so:far} in the sampling process
    described above. Observe that conditioned on $\pi(S)=T'$ the
    partial injection $\pi'^{-1}\cup \tau^{-1}:S\to \VV \setminus S$
    is also uniformly random, and therefore
    \begin{equation*}
        \ProbCond{ (\pi'{}^{-1} \cup \tau^{-1})(S)=T'' 
        }{\pi(S)=T'}%
        \leq%
        \binom{n-s}{s}^{-1}.            
    \end{equation*}
    Hence,
    \begin{align*}
      &\hspace{-0.3cm}%
        \Prob{\pi(S)\cup\pi^{-1}(S)\cup \tau^{-1}(S\cap\pi(S))\subseteq T}
      \\%
      &
        =%
        \Prob{\Bigl.
        \smash{\bigvee\nolimits_{{T'\subset T,\,|T'|=s}}\;
        \bigvee\nolimits_{{T''\subset T\setminus T',\, |T''|=s}}}
        \pi(S)=T' \wedge (\pi'{}^{-1}\cup\tau^{-1})(S)=T''
        }
      \\[0.1cm]
      &%
        \leq 
        \sum\nolimits_{{T'\subset T,\,|T'|=s}}\;
        \sum\nolimits_{T''\subset T\setminus T',\,|T''|=s}
        \Pr \bigl[\pi(S)=T'\bigr] \cdot  \Pr \bigl[(\pi'{}^{-1}\cup\tau^{-1})(S)=T''\,\big|\,
        \pi(S)=T'\bigr]
      \\[0.1cm]
      &%
        =
        \frac{\binom{t}{s}\binom{t-s}{s}}{\binom{n}{s}\binom{n-s}{s}}
        \leq%
        \left(\frac{t}{n}\right)^{s}
        \left(\frac{t-s}{n-s}\right)^{s}
        \leq%
        \left(\frac{t}{n}\right)^{2s},
    \end{align*}
    which complete the proof of \Eqref{pi,pi^-1}.

    The probability on $\{\GS\}_{S\subset \VV}\sim \wGnd$ that there
    exists $S\subset \VV$ of cardinality at most $\eps n$ for which
    $|\Gamma_{\GS}(S)|\leq t$, where $t=(1-\eta)ds$ is at most the
    probability that there exists $\tT\subset \VV \setminus S$ of
    cardinality $t$ for which $\Gamma_{\GS}(S)\subseteq \tT\cup S$.
    By the union bound and~\Eqref{pi,pi^-1} this is at most
    \begin{equation*}
        \sum_{s=1}^{\eps n} \binom ns \binom {n-s}t
        \pth{\pth{\frac
              {t+s}n}^{2d}}^{d/2}.  
    \end{equation*}
    We bound the summand above by
    \begin{align*}
      \binom ns \binom {n-s}t \left(\frac {t+s}n\right)^{sd}
      & \leq \binom ns \binom {n}t \left(\frac {t+s}n\right)^{sd} 
      \\ & \leq \left(\frac{en}s\right)^{s} 
           \left( \frac{en}{(1-\eta)ds}\right)^{(1-\eta)ds}
           \left(\frac{(1-\eta)ds+s}{n}\right)^{ds}
      \\ & \leq \left( e^{(1-\eta)d+1} (1-\eta)d 
           \left( \frac{(1-\eta)d+1}{(1-\eta)d}\right)^d
           \left( \frac{(1-\eta)ds}{n}\right)^{\eta d-1} \right)^s
      \\ & \leq \left( e^d {}d  e^{1/(1-\eta)}
           \left( \frac{ds}{n}\right)^{\eta d-1} \right)^s
      \\ & \leq \left(4 e^d {}d  
           \left( \frac{ds}{n}\right)^{\eta d-1} \right)^s.
    \end{align*}
    
    To bound the above we let $\eps= e^{-1.1/\eta}/d$, and we do a
    case analysis according to the value of $s$.
    
    For $3\log_2 n\leq s\leq e^{-1.1/\eta} n/d$, we have
    \begin{align*}
      \left(4 e^d {}d 
      \left( \frac{ds}{n}\right)^{\eta d-1} \right)^s
      \leq \left( 4e^d {}d  e^{-1.1d}e^{1.1/\eta}\right)^s
      \leq 2^{-s}\leq n^{-2}.
    \end{align*}
    For $1\leq s \leq 3 \log _2 n$, we have
    \begin{align*}
      \left(4 e^d {}d \left( \frac{ds}{n}\right)^{\eta d-1} \right)^s
      &\leq%
        \left(4 e^d {}d \left( \frac{3\log_2 n}{\sqrt{n}}\right)^{\eta
        d-1} \right)^s%
      \\%
      &\leq%
        \left( 4d \left(
        \frac{e^{1/\eta}3\log_2 n}{\sqrt{n}}\right)^{\eta d-1}
        \right)^s \leq \left( d \left(
        \frac{1}{{n}^{6/16}}\right)^{\eta d-1} \right)^s \le
        n^{-2}.
    \end{align*}
    Summing the above over $s=1,\ldots,e^{-1.1/\eta} n/d$, we have
    \begin{equation*}
        \ProbC_{\{\GS\}_{S\subset \VV}\sim  \wGnd} \bigl[
        \forall S\subset \VV, |S|\leq e^{-1.1/\eta} n/d  \implies
        |\Gamma_{\GS}(S)|\ge (1-\eta)d|S| \bigr ]\ge 1-O(n^{-1}).    
    \end{equation*}
    So, conditioned on $\EventB$, by~\Eqref{|G|<|GS|},
    $\forall S\subset \VV, |S|\leq e^{-1.1/\eta} n/d$ we have
    \begin{equation*}
        |\Gamma_G(S)|\ge |\Gamma_{\GS}(S)|-\eta d |S|.    
    \end{equation*}
    
    Since $\Pr_{\Gnd}[\EventB]\ge 1-O(n^{-0.4})$, we conclude
    \begin{equation*}
        \ProbC_{\G\sim \Gnd}
        \Bigl[{
           \forall S\subset \VV, |S|\leq
           e^{-1.1/\eta} n/d%
           \implies%
           |\Gamma_\G(S)| \geq(1-2\eta)d|S|
        }\Bigr]%
        \geq%
        1-O(n^{-0.4}).    
        \SubmitVer{\qedhere}%
    \end{equation*}
\end{proof}

\subsection{Summary}

\RestatementOf{\thmref{proper:expander}}{\ThmProperExapdnerBody}

\begin{proof}
    Property \itemref{p:1} is proved in \lemref{p:1}.  Property
    \itemref{p:2} is proved in \lemref{p:2}.  Property \itemref{p:3}
    follows from \thmref{p:3}.
\end{proof}

\end{document}